\documentclass[a4paper]{article}

\usepackage{algorithm}
\usepackage{algpseudocode}
\usepackage{multirow} 
\usepackage{authblk}

\usepackage{epsfig,shadow}
\usepackage{amsmath}
\usepackage{amssymb}
\usepackage{mdwmath}
\usepackage{amsthm}
\usepackage{mdwtab} 
\usepackage{url}
\usepackage{subfig}
\usepackage{tikz}

\usepackage{draftwatermark}

\newtheorem{principle}{Principle}
\newtheorem{definition}{Definition}

\newtheorem{lemma}{Lemma}
\newtheorem{property}{Property}

\begin{document}

\title{Backtracking algorithms for service selection}

\author[1]{Yanik Ngoko\thanks{yanik.ngoko@lipn.univ-paris13.fr}}
\author[1]{Christophe C\'erin\thanks{christophe.cerin@lipn.univ-paris13.fr}}
\author[2]{Alfredo Goldman\thanks{gold@ime.usp.br}}
\author[3]{Dejan Milojicic\thanks{dejan.milojicic@hp.com}}

\affil[1]{ Laboratoire d'Informatique de Paris Nord of the University of Paris 13, France }
\affil[2]{ Institute of Mathematics and Statistics of the University of S\~ao Paulo, Brazil}
\affil[3]{ HP Labs, USA}

\date{ }

\maketitle

\begin{abstract}
In this paper, we explore the automation of  services' compositions. We 
focus on the service selection problem. In the formulation that we consider, the problem's 
inputs are constituted by a behavioral composition whose abstract services must be bound to  
concrete ones. The objective is to find the binding that optimizes the {\it utility} 
of the composition under some services level agreements. We propose a complete solution. 
Firstly, we show that the service selection problem can be mapped onto a 
 Constraint Satisfaction Problem (CSP). The benefit of this mapping is that the large know-how 
in the resolution of the CSP can be used for the service selection problem. Among the existing 
techniques for solving CSP, we consider the backtracking. Our second contribution is to propose various 
backtracking-based algorithms for the service selection problem. The proposed variants are inspired by  
existing heuristics for the CSP. We analyze the runtime gain of our framework over an intuitive resolution 
based on exhaustive search. Our last contribution is an experimental evaluation in which we demonstrate 
that there is an effective gain in using  backtracking instead of some comparable approaches. The 
experiments also show that our proposal can be used for finding in real time, optimal solutions on 
small and medium services' compositions.
\end{abstract}

{\bf keywords: } {Services' Composition; Service Selection; Business Process; Constraint Satisfaction Problem}

\section{Introduction}

One of the most promising ideas developed in Service Computing is the automatic generation  
of services' compositions~\cite{Papazoglou03service-orientedcomputing,SOCValery,Bartalos}. 
Multiple views and problems formulations have been proposed for this purpose~\cite{Rao04asurvey}. 
As a point of differentiation between these works, we consider the {\it functional structure} 
of the composition to built. By this term, we mean, the representation of the functioning of the composition, as a 
collaboration among small abstractions, that can each be realized by known services. Typically, oriented graphs, 
 business processes and workflows have been used for such descriptions~\cite{Cardoso2004281,Weske2007,GoldmanNgoko}. 
Based on the functional structure, we distinguish between two dominant classes of services' compositions problems. 
In the first one, the general inputs of the composition problem consists of: (1) a basis of services whose behavior is described in the public interface; (2) a set of user constraints and goals, defining and framing the finality of the composition. 
One must infer from these data an interaction among services, meeting the constraints 
and goals~\cite{Bartalos,Sirin:2004:HPW:1741306.1741331,Jiang}. We put these formulations in the class of {\it structurally-unfixed}. 
Their particularity is that  the functional structure of the composition is not an input data for the automation problem. 
It is the case in the second class of formulations that we refer to as the one of {\it structurally-fixed} problems. 
The  challenge is reduced to a binding problem, in which concrete implementations must be associated with the 
abstractions of the functional structure such as to guarantee a minimal Quality of Service (QoS), while 
meeting some users level agreements (SLAs)~\cite{Alrifai,BenMokhtar,Yu,ZengMiddleware,Zheng,Ardagna,JISA,cpe3015}. 
This binding problem is also referred to as the service selection problem.

It seems obvious that structurally-unfixed formulations include more 
automation in the design of services' compositions. Indeed, we can decompose the challenge in these cases 
in two: (1) find a functional structure that meets the constraints and goal; (2) solve a structurally-fixed problem. 
However, let us remark that in practice, it is hard to address structurally-unfixed problems 
without providing a formal description of the composition's behavior. The OWL-S language~\cite{Sirin:2004:HPW:1741306.1741331} and pre/post condition formalisms~\cite{Bartalos,Oh05acomparative} are some examples, utilized in this context.
The existence of these additional inputs in practice tempers the high level of automation of structurally-unfixed 
formulations. 

In this work, we consider the service selection problem (structurally-fixed formulation). The functional 
structure in our work is given by a Hierarchical Services Graph (HSG)~\cite{GoldmanNgoko}. This modelling defines 
a  composition as a graph with three layers: the machine, service and operations layers. The service composition 
logic is captured in the operations layer. The logic consists of BPMN~\cite{Weske2007} interactions 
among a set of operations. These operations are abstract in the sense that they can be 
implemented by different services located in the underneath layer. Given such a graph, 
we are interested  in finding the {\it best} implementation for abstract operations while fulfilling the SLAs constraints. 
We restrict the SLAs definition to two QoS dimensions:  the Service Response Time (SRT) and the Energy Consumption (EC). 

Although we use a particular representation of services' compositions, the core problem that we address is not 
new~\cite{Alrifai,BenMokhtar,Zheng,Ardagna,JISA}. But, our study has two main features. Firstly, we are interested in 
finding optimal solutions. This choice has a weakness: the NP-hardness of the service selection problem. However, 
we believe that by making an  {\it intelligent search}, one can provide, within an acceptable runtime,  
exact solutions for {\it small or medium} services' compositions (around $20$ nodes for the service composition). 
It is important to notice that if we consider that the services compositions implement business processes or workflows, 
then there are many practical examples that will correspond to small or medium compositions. One can look 
for instance to the examples provided in~\cite{Omg,Freund}. 
The second feature of our work is that we adopt a view 
of the problem that has not been studied to the best of our knowledge. This is clearly stated by our 
main contribution, which consists of mapping the service selection problem on the Constraint Satisfaction 
Problem (CSP). This mapping opens a 
large potential in the resolution of the service selection problem. As we will see, 
it also captures another facet of service selection: the feasibility problem~\cite{Ardagna,JISA}. 

Among the existing
techniques for solving CSP, we choose to investigate the 
backtracking~\cite{Baker95intelligentbacktracking}. We complete our contribution in proposing various 
backtracking-based algorithms for the service selection problem. The proposed variants are based on the notion of 
reduction order, introduced in our prior work~\cite{GoldmanNgoko} on QoS prediction. They are also inspired by two 
existing heuristics~\footnote{The referred heuristics provide an exact solution; but, their runtime can differ significantly 
from an instance to another} for the CSP:  the max-degree and min-domain heuristics~\cite{Baker95intelligentbacktracking}. 
Finally, we did an experimental evaluation in which we demonstrate the runtime gain of the backtracking-based 
algorithms over two classical solutions for solving the service selection 
problem: exhaustive search and integer linear programming. The experiments also give us interesting insights about 
the size of the compositions on which we can expect results in real-time. 

The remainder of this paper is organized as follows. Section~\ref{Related} presents the related works. 
In Section~\ref{SSP}, we define the service selection problem and connect it to the CSP. Naive and backtracking 
algorithms for the problem are discussed in the sections~\ref{exhaustiveSearch} and~\ref{backtrackingSearch}. 
Section~\ref{ExperimentalEvaluation} gives an experimental evaluation of our work. In Section~\ref{Discussion}, we 
discuss about the potential in the development of service selection algorithms raised by our CSP mapping.  We 
conclude in Section~\ref{Conclusion}.

\section{Related work} \label{Related}

We assumed that services' compositions problems can be structurally-unfixed or structurally-fixed. 
Our work will focus only on the latter case. However, interesting 
proposals for the former case can be found in the work of Seog Chan et al.~\cite{Oh05acomparative}, 
Sirin et al.~\cite{Sirin:2004:HPW:1741306.1741331} or in the survey of Rao and Su~\cite{Rao04asurvey}. A key idea 
that these works share is the usage of AI planning techniques for the resolution of the services' composition problem. 

Concomitantly, the Integer Linear Programming (ILP) is one of the most used techniques employed for tackling the service 
selection problems. One of the pioneer papers with this technique was done by Lee~\cite{Lee}. 
Though it was not its main purpose, his work demonstrated that SLAs can, in practice,  
be formulated as linear equations. This suggests that in many cases, the service selection problem can 
be solved by integer linear programming. The work of Lee focused only on two QoS dimensions: the price and the 
service response time. Similar modellings were proposed including other QoS dimensions like the
price, duration, reputation, reliability and availability~\cite{ZengMiddleware,Yu,Ardagna}. In particular, 
the work of Zeng et al.~\cite{ZengMiddleware} shows that if we consider their natural formulation, 
constraints related to availability will result in non-linear equations. Then, they propose a method for 
transforming such equations into linear ones. In most papers based on linear programming, 
the services' composition is viewed as a collaboration among multiple services within a single 
business process. In practice however, collaborations among multiple processes exist. For these latter 
cases, the works of Ngoko et al.~\cite{JISA,mgc2012} stated how to use linear programming for the 
service selection problem. They focused on energy and service response time minimization.

Many papers established in different contexts the NP-hardness of the service selection problem~\cite{Lee,Yu}. 
This means that exact solutions obtained by Integer Linear Programming (ILP) are computed 
in exponential runtime. 
For obtaining fast results, heuristics have  been considered. 

Yu et al.~\cite{Yu,Yu:2004:SSA:1018413.1019049} proposed a branch and bound algorithm (BBLP) and a dynamic 
programming algorithm  for service selection. These ideas are also discussed in~\cite{Lee}. In both cases the service
 selection problem is reduced to the multichoice 
knapsack problem. The branch and bound algorithm exploits the lagrangian relaxation of the ILP modelling 
for this problem. The dynamic programming solution adapts existing solutions for the multichoice knapsack problem. 
BBLP and dynamic programming  improve in runtime the naive resolution of the service selection problem. 
In this work, we will also propose an exact resolution (as the branch and bound solution of Yu et al.) that 
improves the naive one.
Ben Mokhtar et al.~\cite{BenMokhtar} proposed a two-phases 
heuristic for service selection. In the first phase, the heuristic classifies services regarding 
their {\it utility}. A search approach based on this classification completes the selection process in the second phase. 
The advantage of this approach is to propose near-exact solutions within {\it short} runtime. Ngoko et al.~\cite{JISA} 
and Zeng et al.~\cite{ZengMiddleware} proposed sub-optimal 
approaches and exact algorithms on special cases for the resolution of the service selection problem. As they stated, these 
solutions are efficient only on a small class of services' compositions. Yu et al.~\cite{Yu:2004:SSA:1018413.1019049} 
proposed  {\it improvement heuristics} for finding near optimal 
solutions. The first step of the heuristic consists of finding a feasible solution. 
The solution is then improved by considering other potential combinations of services. As they showed, the proposed 
heuristic can improve the runtime required for {\it BBLP}. Let us however recall that solutions found are not 
necessarily optimal. Alrifai et al.~\cite{Alrifai} also proposed a heuristic for the service selection problem. 
The main idea is to decompose the global service selection problem onto subproblems that can be solved by local 
search. Doing so, they show that they can drastically improve the runtime of the heuristic proposed by Yu et al.~\cite{Yu}. 
Est\'evez-Ayres et al. ~\cite{DBLP:journals/concurrency/Estevez-AyresGBD11} proposed a heuristic 
for finding good approximations for the service selection problem under a runtime boundary. One of the main 
ideas that they describe will be used in our work. It consists of making partial evaluation on a sub-composition of 
services. In our proposal, we go deeper in the idea; we focus on the ordering that can be used for the 
set of partial evaluations to be made. Genetic programming approaches were also 
proposed~\cite{Jaeger,Canfora,Cao,cpe3015} in viewing a services' composition as a string where each character is a service 
operation. Finally, let us remark that as suggested by Cardoso et al.~\cite{Cardoso2004281} and developed in other 
works~\cite{DBLP:journals/concurrency/Estevez-AyresGBD11,Yu:2004:SSA:1018413.1019049}, heuristics for the service selection 
can be obtained from those solving the QoS prediction problem. We will come back on this point later. For now, 
let us point out that many QoS 
prediction algorithms as the SWR algorithm~\cite{Cardoso2004281} or other graph reduction algorithms~\cite{GoldmanNgoko,Zheng2} 
can serve for designing heuristics for the service selection problem.

A large literature exists about the service selection problem. Existing works established connections between the service selection problem 
and the SAT problem~\cite{Oh05acomparative}, the multichoice Knapsack problem~\cite{Lee}, the multidimension multichoice Knapsack problem and the 
multiconstrained optimal path problem~\cite{Yu:2004:SSA:1018413.1019049}. However, we did not find any work that proposes to study this problem 
as a variant of the CSP. The solution that we propose is inspired by a CSP view of service selection.  It also has some similarities with the work 
done by Est\'evez et al.~\cite{DBLP:journals/concurrency/Estevez-AyresGBD11} the branch and bound algorithm of Yu et al.~\cite{Yu} and 
our prior work~\cite{JISA}. For circumscribing our originality, let us also notice that our prior work was based on mixed integer linear 
programming. Though interesting, this solution is tuned to a particular QoS modelling. For instance, it does not work with 
the probabilistic modelling of Hwang et al.~\cite{Hwang20075484}. As Est\'evez et al.~\cite{DBLP:journals/concurrency/Estevez-AyresGBD11}, we 
propose to improve the runtime of the service selection algorithm by reducing the set of candidate solutions that are explored. We differ 
from their work by  the usage of the backtracking technique. 
Similarly to Yu et al.~\cite{Yu}, we adopt a branching technique for reducing the set of candidate solutions 
during the resolution. But instead of using the lagrangian relaxation, we propose a novel estimation.
Finally, let us observe that while most work focused on heuristics approaches for the service selection problem, we 
are interested in exact resolution. The NP-hardness of the service selection problem is the main weakness of this 
choice. However, we believe that in using an appropriate search algorithm, one can obtain in a short 
runtime optimal solutions for {\it small and medium services'  compositions}. Our work will demonstrate this point in considering 
services' compositions implementing well-known business processes. Moreover our proposal can serve as a solid basis for the 
development of approximated solutions on the service selection problem.

\section{The service selection problem} \label{SSP}

\subsection{Structure of a services' composition}

We model a services' composition as a Hierarchical Services Graph (HSG)~\cite{GoldmanNgoko,mgc2012,JISA}. 
In this representation, a services' composition is a three-layers graph, each (of the layers) encapsulating a particular abstraction of a service 
composition; these are the business processes, the services and the physical machines layers. 
An example of such a graph is given in Figure~\ref{ExampleHSG}. 
\begin{figure}[htbp]
\centering
\fbox{
\includegraphics[width=0.68\linewidth,height=2.1in]{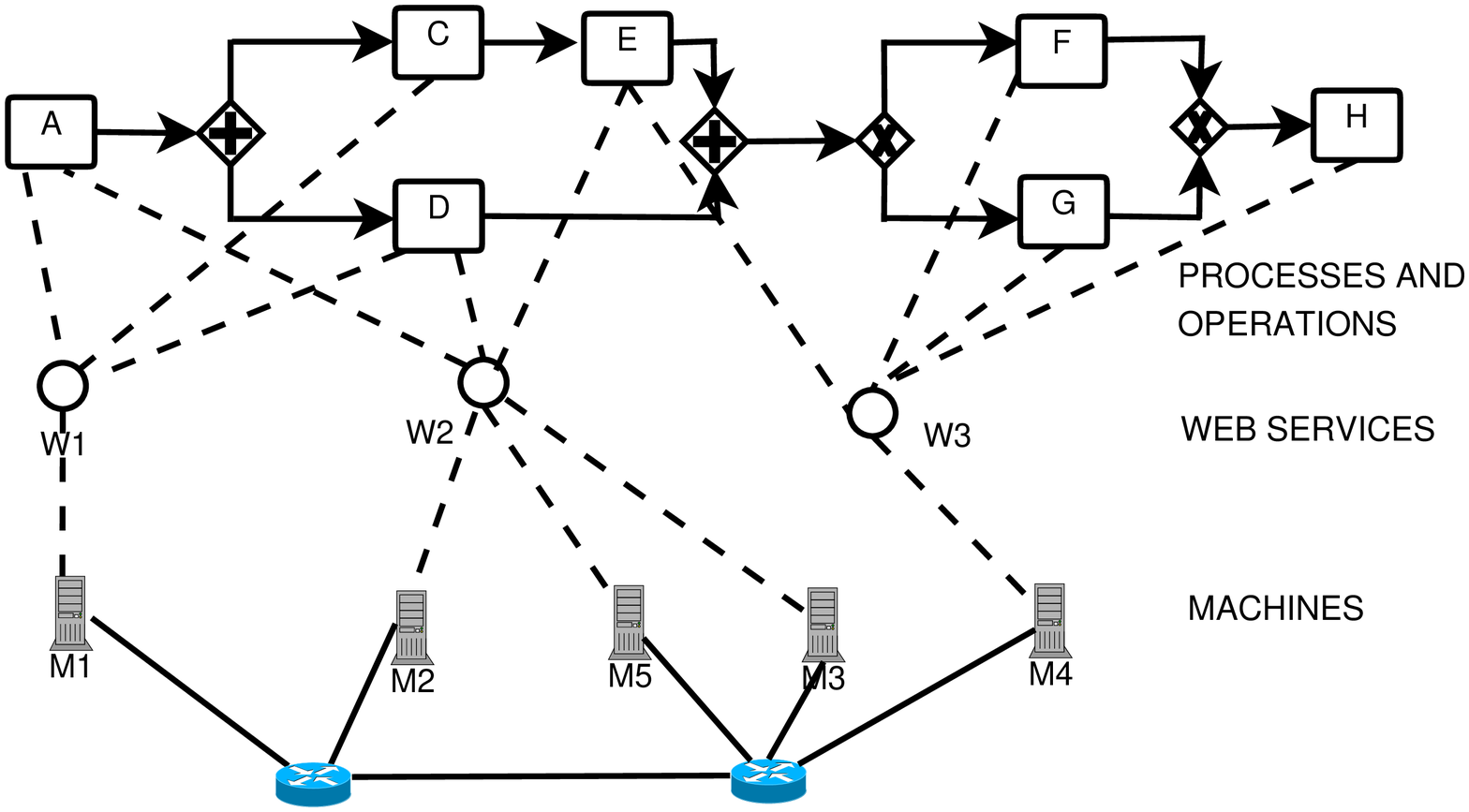}}
\caption{Example of Hierarchical Services Graph}\label{ExampleHSG}
\end{figure}

\begin{figure}[htbp]
\centering
\fbox{
\includegraphics[width=0.64\linewidth,height=2.4in]{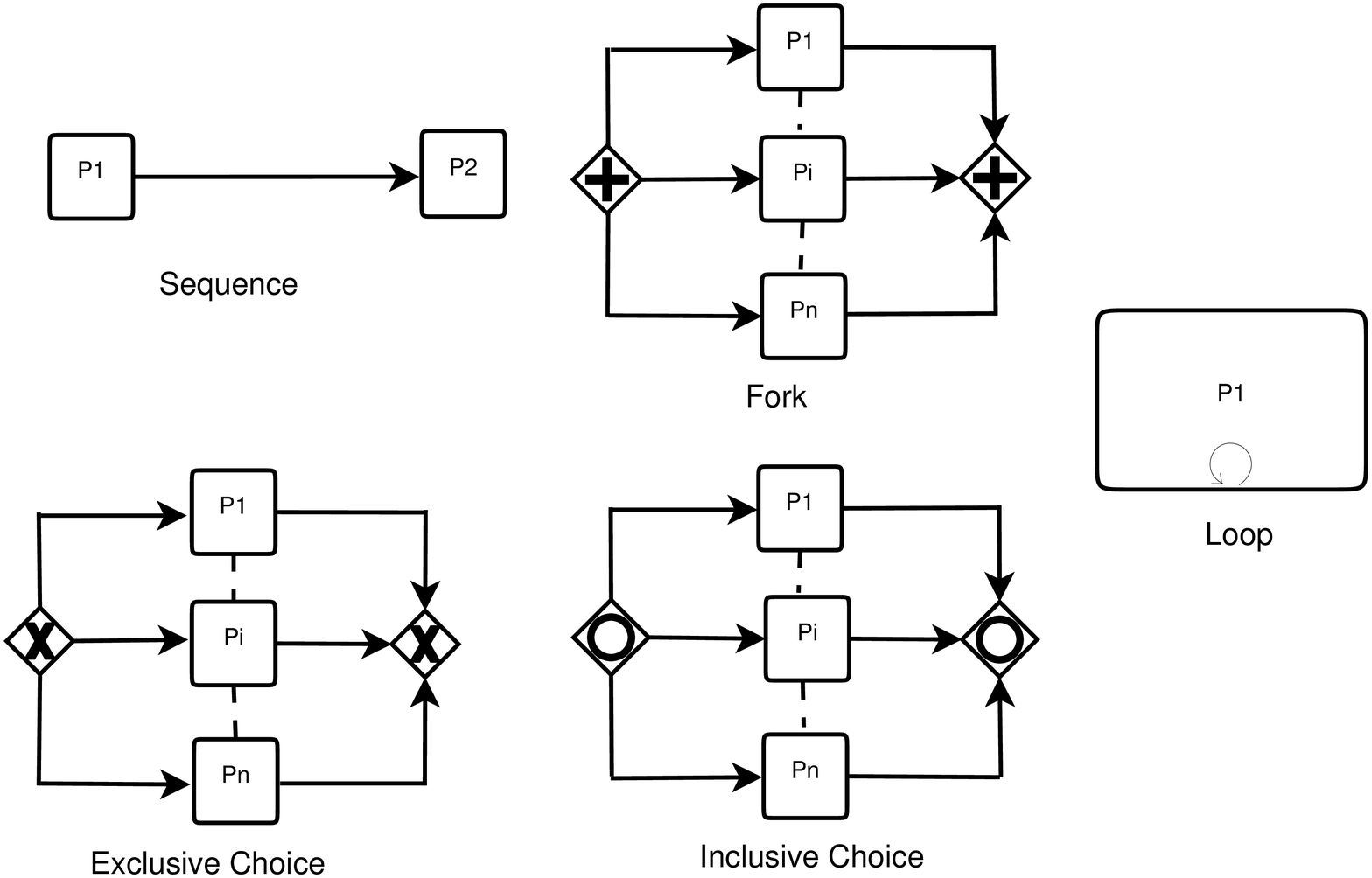}}
\caption{Set of generic subgraph patterns in the operations layer of a HSG. Each $G_i$ 
is again a subgraph obtained from these patterns or an operation}\label{flows}
\end{figure}

A HSG comprises three layers organized as follows. The first layer, which we will also refer to as the \textbf{operations graph}, 
describes the functioning of the services' compositions as a business process interaction among abstract operations. 
For the sake of simplicity, we will reduce these interactions to the ones obtained by composing the subgraphs patterns of Figure~\ref{flows}. 
Abstract operations are implemented in the services that are in the layer underneath. Finally the last layer states the machines where 
services are deployed. 

An operation can be implemented in multiple services. This means that given an abstract operation, we have at the services 
layer, many concrete operations that can implement its behavior. A same service can be deployed on various machines. This accounts 
for capturing the possible migrations that can occur during the execution of a services' composition. However, for the sake of simplicity, 
we will assume that each service is deployed on a unique machine. 
Finally, let us remark that the HSG modelling we consider corresponds to the implementation of a single business process. 
Representations for multiple business processes also exist~\cite{mgc2012,JISA}; but, this is out of the scope of our 
study.

The service selection problem is based on predefined relationships between abstract operations and services. 
The problem consists in choosing the best concrete operations for abstract ones such as 
to minimize the service response time and the energy consumption. Below, we provide a formal definition. 

\subsection{Problem formulation}

Here, we consider the formulation introduced in~\cite{JISA,cpe3015}

{\bf Problem's inputs: }
Let us consider a HSG whose set of abstract operations is $O$.  For each operation $u \in O$, 
there is a set of concrete implementations $Co(u) = \{u_1, \dots, u_{m_u}\}$. For each concrete 
implementation $u_v$, we have the mean response time $S(u_v)$ and the energy consumption $E(u_v)$. 
Finally, we assume two upper bounds that are issued from SLAs constraints: the bound $MaxS$ for 
the service response time and $MaxE$ for energy consumption. Finally, we have a tuning 
parameter $\lambda \in [0,1]$ used  for giving more priority to the SRT or the EC in the problem 
optimization goal.\\

{\bf Problem objective: }
We are looking for an assignment of concrete operations for $O$ that fulfills the following constraints:
\begin{enumerate}
\item[$C_1$:] each operation must be associated with a unique concrete implementation;
\item[$C_2$:] the QoS of the resulting composition must not exceed $MaxS$ in response time and  $MaxE$ in energy consumption;
\item[$C_3$:] if $S$ is the service response time of the resulting composition and $E$ its energy consumption, then 
the assignment must minimize the global penalty $\lambda S + (1-\lambda)E$. 
\end{enumerate}

In this formulation, the constraint $C_2$ defines users' SLAs for the response time and 
energy consumption. The composition has a penalty defined by the constraint $C_3$. 
For completing this formulation, it is important to explain how $S$ and $E$ are computed given 
a binding of abstract services to concrete ones. We address this issue in the following subsection by associating an 
execution semantics to HSGs. 

\subsection{Execution semantics}

We divide the semantics in two parts. The first one states how we represent the QoS of a concrete 
operation; the second one determines how we compute the QoS of a request that {\it traverses} multiple 
concrete operations.

\subsubsection{QoS of a concrete operation.}

We use a deterministic modelling for QoS operation. The mean of the QoS of a concrete 
operation in a given dimension (response time, energy consumption) is expressed as a real value. 
Though criticizable, this modelling has been considered in multiple other works~\cite{Cardoso2004281,ZengMiddleware,GoldmanNgoko}. 
Moreover, the conclusion of our current work can be extended to other modellings like the  
probabilistic one  of Hwang et al.~\cite{Hwang20075484}. Given the QoS of each concrete operation, we will now show how to 
aggregate them in order to capture the mean QoS of a request that is processed in a HSG.

\subsubsection{QoS of a subgraph.}

We compute the QoS of a service composition depending on the structure of the operations 
graph (upper layer graph of the HSG). The idea is to aggregate the operation QoS, in considering all possible execution cases 
of a request processed in a HSG. The aggregation rules are depicted in Table~\ref{tabAggRules}. 
They state how to compute the response time and the energy consumption,  
expected from a request that is processed in a HSG whose structure 
is matched to the patterns of Figure~\ref{flows}.
$E(P)$ refers to the energy consumption of a request processed in the subgraph $P$. 
$S(P)$ is its response time. In {\it exclusive choice},  $p_i$ gives the probability for a request to be routed towards the subgraph $P_i$.  
For the sake of simplification in {\it inclusive choice} (see Figure~\ref{flows}), we  assume  that the request can only be routed towards the subgraph 
$P_1$, $P_2$ or simultaneously, to both. Each routing occurrence, has  
a known probability $p_{or1}$, $p_{or2}$ and $p_{or||}$. We added this restriction on inclusive choices for the sake of 
simplicity. Our solutions can however be generalized to the case where we have more routing occurrences. 
Finally, for any loop subgraph, we assume that we have 
the mean number $m$ of times in which the requests loop on it. 
\begin{table}[htbp]
\centering
\begin{tabular}{|p{4cm}|p{4cm}|p{2cm}|} 
\hline
\small \textbf{Sequence}   & \small \textbf{Fork}  & \small \textbf{Loop} \\\hline
$S(P_1) + S(P_2)$ & $\max\{ S(P_1),\dots S(P_n) \}$ & $m.S(P_1)$ \\
$E(P_1) + E(P_2)$ & $ E(P_1)+ \dots+ E(P_n)$ & $ m.E(P_1)$ \\\hline
\small \textbf{Exclusive choice} &  \multicolumn{2}{c|}{\small \textbf{Inclusive choice}} \\\hline
 $\sum_{i=1}^n p_i.T(P_i)$  & \multicolumn{2}{c|}{$p_{or1}.T(P_1)+ p_{or2}.T(P_2)$  $ + p_{or||}.\max \{T(P_1), T(P_2)\}$}  \\
 $\sum_{i=1}^n p_i.E(P_i)$  & \multicolumn{2}{c|}{$p_{or1}.E(P_1)+ p_{or2}.E(P_2) +  p_{or||}.(E(P_1)+ E(P_2))$}  \\\hline
\end{tabular}
\caption{Aggregation rules on subgraphs patterns}\label{tabAggRules}
\end{table}
\normalsize

In these formula, we have almost the same aggregation rules for energy consumption and 
response time. The difference between these two dimensions is on how we interpret the parallelism: 
from an energy viewpoint, all paths of execution, {\it even parallel}, will induce an energy consumption. 
For additional explanations about these formula, we refer the reader to~\cite{JISA}. 

From the Lee result~\cite{Lee}, it is easy to establish that the described service selection 
problem under our execution semantics can be reduced to a multi-choice knapsack problem; this  
proves its NP-hardness. Below, we will show that we can use the constraint satisfaction problem 
for solving the service selection problem.

\subsection{Service selection as a constraint satisfaction problem} \label{Decomposition}
The Constraint Satisfaction Problem (CSP) is a classical problem in artificial intelligence and combinatorial optimization. 
A CSP is defined as a tuple $(V, D, C)$
where:
\begin{itemize}
\item $V = \{ v_1,\dots, v_n \}$ is a set of variables;
\item $D = \{ D(v_1),\dots, D(v_n)\}$ is the set of variables' domains;
\item $C = \{C_1,\dots, C_m\}$ is the set of constraints; each $C_i$ imposes a restriction on the possible 
values that can be assigned to a subset of variables; 
\item There are two classical objectives in this problem. In the {\it one-solution} objective, we are 
looking for an assignment of values to variables that does not violate any constraint. In the {\it all-solutions}  
objective, we are looking for all assignments that do not violate the constraints. 
\end{itemize}

Regarding the objective function, we will show that the  {\it one-solution} objective captures the 
feasibility problem in service selection~\cite{Ardagna,JISA} while the {\it all-solutions} objective captures the 
service selection problem. We recall that in the feasibility problem problem, the interest is in 
finding an assignment of concrete services that meet the SLAs.  

Firstly, let us assume the {\it all-solutions} objective. 
Let us also assume that we have to solve the service selection problem for an arbitrary HSG. We propose to map 
the problem onto a  CSP throughout the following rules:
\begin{enumerate}
\item we consider that variables correspond to the operations of the HSG;
\item the domain of a variable is the set of possible concrete operations that 
implements the abstract one to which it refers to;
\item the constraints of the problem are the SLAs of the service selection problem. This means that 
we are looking for all assignments $f \subseteq  D(v_1)\times \dots \times D(v_n)$, such 
that $E(f) \leq MaxE$ and $S(f) \leq MaxS$; here $E(f)$ and $S(f)$ are the 
energy consumption and the response time.
\end{enumerate}

The resolution of this problem will return all candidate solutions for the service selection problem. 
Let us suppose that it gives us $\omega$ assignments $f_0, \dots, f_{\omega -1}$. For solving the service 
selection problem, we select the assignment $f_{opt}$ such that 
$\displaystyle opt = \arg \{ \min_{0\leq u \leq \omega-1} \lambda.S(f_{u}) + (1-\lambda).E(f_u) \}$. 

CSPs are often classified according to the constraints formulation~\cite{Baker95intelligentbacktracking}. 
In binary CSPs for instance, each constraint is defined as a set of pair values that cannot be associated 
with two specific variables; this can be generalized to $k$-ary CSPs where constraints are defined as 
tuples of $k$ values that cannot be associated with $k$ distinct variables. In nonlinear CSPs, constraints are 
formulated as nonlinear inequalities on variables. With the proposed mapping, 
the service selection problem is a nonlinear-like CSP. 

It is straightforward to notice that in 
applying the given mapping with the {\it one-solution} objective, we obtain the feasibility problem 
in service selection. 
The proposed mapping can be extended to many other formulations of the service selection problem. 
For instance one can include in it other SLAs constraints on reputation, price, or availability. 
One of its main benefits is to suggest that the service selection problem can be solved by adopting CSPs 
algorithms. For this, we need to provide an {\it evaluation algorithm} that states how given an 
assignment $f_u$, we compute $E(f_u)$ and $S(f_u)$. 
In the following text, we describe this algorithm and a first solution for the service selection problem

\section{Evaluation algorithm and exhaustive search} \label{exhaustiveSearch}

\subsection{Evaluation algorithm} \label{evaluationAlgorithm}

We propose to use our prior QoS prediction algorithm~\cite{GoldmanNgoko}. 
We will recall below some key points of the evaluation algorithm.

In the algorithm proposed in~\cite{GoldmanNgoko}, the QoS are computed with the graph reduction technique. 
We consider as input a HSG whose operations graph that is obtained by composing the patterns of Figure~\ref{flows}. We 
will say that such a graph is {\it decomposable}. The algorithm will successively seek in the operations graph,  
subgraphs whose structure is defined as in Figure~\ref{flows}, but with the $P_i$s, corresponding 
here to operations. We will use the term \textit{elementary subgraphs} for qualifying them. 
As soon as an elementary subgraph is found, it is reduced. This means that its QoS 
are computed and the subgraph is replaced by a single node with the same QoS. Then, the execution 
continues until the reduction of the operations graph reaches a single node.

For optimizing the algorithm's runtime, a reduction order is computed at the beginning. The reduction order is a stack of 
subgraphs such that: (1) the top subgraph is elementary, (2) as soon as the top subgraph is reduced, 
the new one in the  top is also elementary. The reduction is done 
according to this order. Goldman and Ngoko~\cite{GoldmanNgoko} showed that elementary subgraphs 
can be characterized by two frontier-nodes: a root and a leaf one. This fact eases the 
subgraphs' representation  in the reduction order. 

In Figure~\ref{Reduction-w-order}, an illustration of the reduction process is provided. 
Initially, we have the graph of Figure~\ref{Reduction-w-order}(1). The first phase of the 
algorithm will generate the reduction order (or the reduction stack) for the graph. We 
represent it as a stack in which subgraphs are given by a root and a leaf node. The second phase begins with 
the unstacking of the top element in the order and then the 
reduction of the corresponding elementary subgraph. This leads to the graph of 
Figure~\ref{Reduction-w-order}(2). The algorithm continues in the same way until 
the reduction stack is empty. At this step, the operations graph will be reduced to a unique node. 
\begin{figure}[htbp]
\centering
\fbox{
\includegraphics[width=0.8\linewidth,height=1.7in]{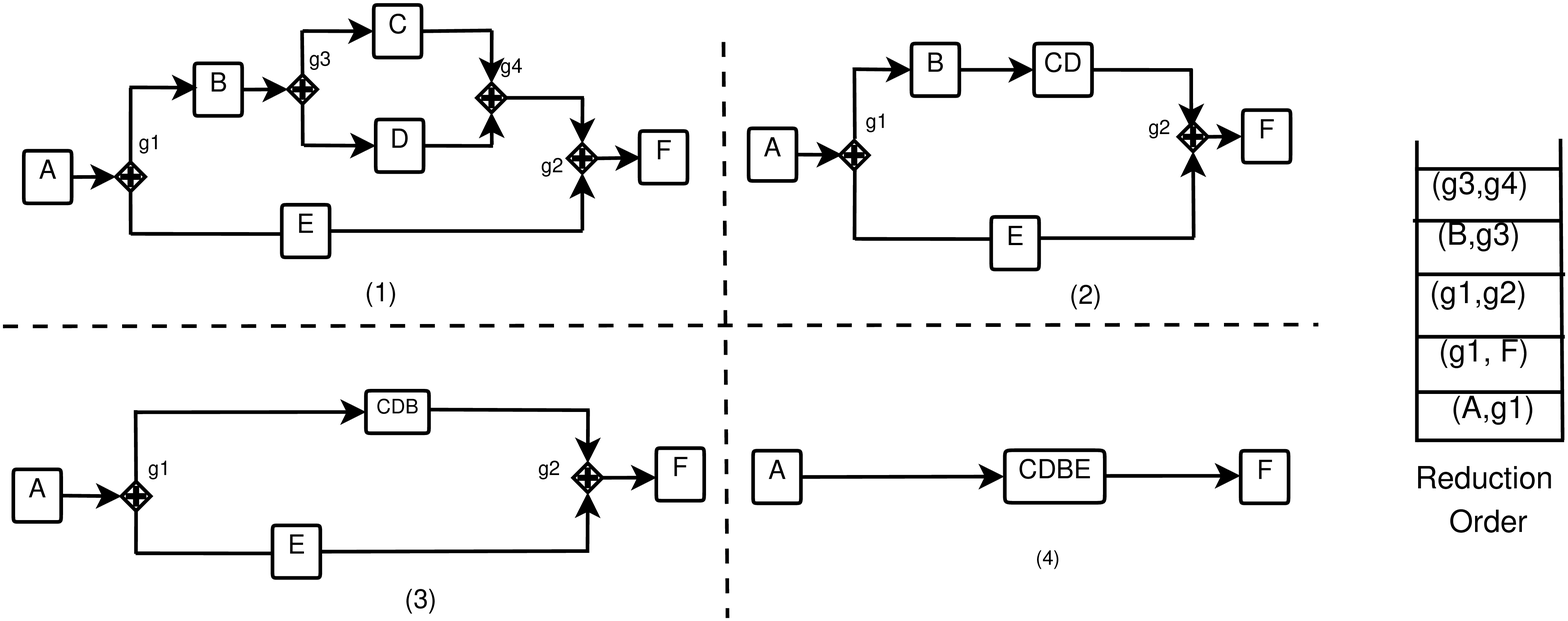}}
\caption{Example of graph reduction.}\label{Reduction-w-order}
\end{figure}

As stated in the 
introduction, our resolution of the service selection problem will be based on the notion of reduction order.
We will recall here some important details about its representation. 

The reduction order is made of pairs $(x,y)$ that each defines a subgraph to reduce. Regarding the composition of 
each pair, four cases can be distinguished: a) $x$ and $y$ are operations;  
in this case, the referred reduction is an elementary sequence with $x$ and $y$; in Figure~\ref{Reduction-w-order}(2) 
for example, we have the reduction $(B, CD)$; b) in the second case, $x$ 
is a split connector and $y$ is a join connector (for instance  $(g_3, g_4)$ in Figure~\ref{Reduction-w-order}(1));  
then, the reduction refers to the elementary split/join subgraphs; c) in the third case, $x$ is a split connector 
and $y$ is an operation; then, the reduction 
refers to the subgraphs whose root is $x$ and leaf is $y$; in Figure~\ref{Reduction-w-order}(1), we have the reduction $(g_1, F)$ that
refers to the subgraph comprising $g_1, B, g_3, C, D, g_4, E, g_2, F$; d) in the last case, $x$ is an operation and 
$y$ a split connector; then, the reduction refers to the subgraph whose root is $x$ and leaf is {\it the leaf of $y$}. 
For instance, $(B, g_3)$ comprises $B, g_3, C, D, g_4$.

Now that the evaluation algorithm has been detailed, we can derive a service selection algorithm 
by stating how we solve the CSP. One can envision at this stage to use a generic CSP solver. 
But, let us observe that in our CSP mapping (See Section~\ref{Decomposition}), we cannot easily map the 
description of the SLAs constraints ($E(f_0) \leq  MaxE$ and 
$S(f_0) \leq  MaxS$) to classical options used in CSP solver (e.g set of unauthorized values, linear equations).  
This is why in the sequel, we will consider a proper resolution.

\subsection{Exhaustive search algorithm}

We propose to consider the exhaustive search algorithm for the CSP~\cite{Baker95intelligentbacktracking,DBLP:journals/concurrency/Estevez-AyresGBD11}. Given a CSP $(V, D, C)$,  the principle of this algorithm is to randomly 
generate assignments of values taken in $D$ to the variables $V$. Each time that an assignment $f$ is generated, 
one evaluates whether or not it fulfills the constraints $C$. If it is the case, we return $f$ as 
solution; otherwise, we generate another assignment. 
In using this algorithm in our mapping of the service selection problem (see Section~\ref{Decomposition}), 
we obtain the Algorithm~\ref{alg:Exhaustive} for the service selection problem. 

The proposed scheme is based on the following notations:
\begin{itemize}
\item $|f|$ is the number of abstract services.
\item With each abstract operation, we associated a distinct integer number (the abstract services have the numbers:  
$0, \dots, |f|-1$). 
\item For defining assignments, we use the array $f$. $f[i]$ will denote the concrete operation  
associated with the abstract operation $i$.
\item $E(f)$ and $S(f)$ are the energy and the service response time of the partial or 
concrete assignment made in $f$.
\item The over-defined notion $Co(index)$ refers to the set of concrete services that can be assigned to 
the abstract operation $index$.
\end{itemize}

Though we deduced the exhaustive search algorithm from the know-how in CSP resolution, let us observe 
that the proposal can be found in other work~\cite{DBLP:journals/concurrency/Estevez-AyresGBD11,Yu:2004:SSA:1018413.1019049}. 
The main difference between our solution and theirs is the evaluation algorithm. 

The exhaustive search proposes an exact solution for the service selection problem. However, considering 
the way we solve the CSP, this solution is not necessarily the best for the runtime viewpoint. 
In what follows, we propose a faster algorithm that includes more in depth the know-how in CSP's resolution.

\scriptsize
\begin{algorithm}[H]                    
\begin{algorithmic}[1]
\scriptsize
\Function{Main}{}
        \State $OptPenalty = +\infty$;  $index = 0$;
	\State Create an uninitialized array $f$ of values for abstract operations;
	\State Call exhaustive($f$, $H$, $OptPenalty$, $index$);
	\State Return the best assignment and $OptPenalty$;
\EndFunction

\Function{exhaustive}{$f$, $H$, $OptPenalty$, $index$}
\If{$index = |f|$}
	\State Compute $E(f)$ and $S(f)$ from the evaluation algorithm with $H$ and $Q$;
	\If{ $S(f) \leq MaxS$ and $E(f) \leq MaxE$ }
		\If{$\lambda.S(f) + (1-\lambda).E(f) < OptPenalty$}
			\State Save $f$ as the best assignment;
			\State $OptPenalty = \lambda.S(f) + (1-\lambda).E(f)$;
		\EndIf
	\EndIf
	\State Return;
\EndIf
\For{ all concrete operations $u \in Co(index)$}
\State $f[index] =$ $u$;
\State Call exhaustive($f$, $H$, $OptPenalty$,  $index+1$);
\EndFor
\EndFunction

\normalsize
\end{algorithmic}
\caption{\scriptsize SS-Exh (Exhaustive search for service selection). \\ {\bf INPUT:} a HSG $H$ and a QoS matrix $Q$ giving the energy consumption and 
service response time of each concrete operation;  \\ {\bf OUTPUT:} An assignment of concrete operations to abstract ones }
\label{alg:Exhaustive}  
\end{algorithm}
\normalsize

\section{A Backtracking search for the Service Selection Problem} \label{backtrackingSearch}

Our objective is to improve the runtime exhaustive search algorithm. Our conviction is that 
this algorithm has an amount of {\it useless work} that can be avoided. This section is 
organized in two parts. In the first one, we discuss about useless work in the exhaustive search. Then, we 
propose an algorithm for avoiding them. 

\subsection{Useless work in exhaustive search}

Our prior work~\cite{JISA} highlights a critical situation that happens in the service selection 
problem: {\it the infeasibility problem}. Indeed, given a HSG $H$, it might be impossible to respect the 
constraints defined in the SLAs  because for a sub-HSG $H' \subset H$, the service selection 
problem does not have any solution. As illustration, let us consider the service selection problem with the 
operations graph of Figure~\ref{noSolution}. In the SLAs constraints, if it is set that the service response time must 
be lower than $13$ms, then the infeasibility of the problem can be established in considering the 
possible assignments for the subgraph $H' = (g_1, g_2)$. 
The exhaustive search here is not always optimal regarding the amount of work. Indeed,  
a better search would have consisted of exploring the possible assignments that can be made for the 
abstract operations in $H'$ and then checking each time whether or not these assignments respect the SLAs.  
Doing so, the infeasibility  could have been established in exploring only a part of the 
search space. 

\begin{figure}[htbp]
\centering
\fbox{
\includegraphics[width=0.9\linewidth,height=1.7in]{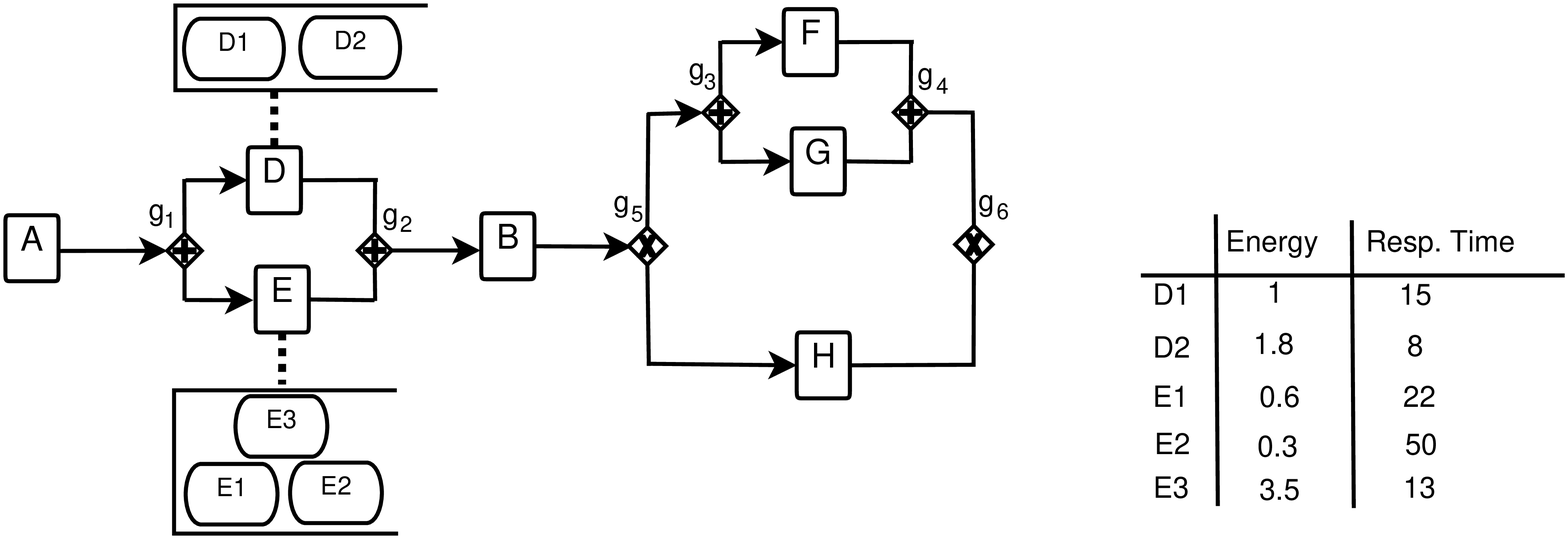}}
\caption{Example of operations graph with related concrete operations. 
$D$ is implemented by $D_1$ and $D_2$ and $E$ is implemented by $E_1, E_2, E_3$}\label{noSolution}
\end{figure}

The second instance of useless work  is similar to the first one. We suppose 
now that the problem is feasible; but, there are multiple assignments that do not respect the 
SLAs. If the constraints violation were already identified in a sub-HSG $H' \subset H$, a part of the 
useless work could have been avoided. As illustration, if we consider in 
Figure~\ref{noSolution} that the response time must be lower or equal to $14$ms, then there is 
only one assignment to the subgraph $(g_1, g_2)$ that can lead to a feasible solution. 
It is only this assignment that must be joined with the other possibilities for $A$ and $B$.

The last situation of useless work is related to the quality of a partial assignment. 
Let us assume that we already have a correct assignment whose penalty is $p$. It might happen in the 
exhaustive search that we made an assignment to a sub-HSG $H'$ whose total penalty already exceed the value of $p$. 
In these cases, we must not try to {\it complete} this assignment (to all operations of $H$) since we already have a 
better one. Let us observe that this analysis is often done in branch and bound algorithms. In the 
context of service selection, a discussion can also be found in the work of Yu et al.~\cite{Yu}.

Summarizing, we have useless work in the exhaustive search if we can find a sub-HSG for which 
multiple assignments are infeasible or already dominated by an existing solution. 
For optimizing these situations we propose to use a backtracking search that we discuss in the following. 

\subsection{Backtracking algorithms}

We consider the CSP resolution with the backtracking technique applied with a static initial ordering. 
Given a CSP tuple $(V, D, C)$, the technique starts by defining an ordering of the 
variables $V$. Let us assume that the resulting order is $v_1, \dots, v_n$. Then, one successively 
assigns values to the variables according to the ordering and their domain definition. In the processing, 
we can reach a situation where values are assigned to  
$v_1, \dots, v_i$. Then, one checks if no constraint is violated by this partial assignment and if the 
assignment is not already dominated by another one. 
If it is not the case, one assigns a value to $v_{i+1}$. Otherwise, one will assign another 
value to $v_i$ (one backtracks). 

The backtracking technique might reduce the useless work that we have identified before. In 
exhaustive search, we evaluate all possible assignments to the variables. 
In backtracking, this is not the case. If for instance, there is no assignment for $v_1$ that 
satisfies the constraints, then backtracking will consider only $|D(v_1)|$ assignments instead of 
$|D(v_1)|\times \dots \times |D(v_n)|$ for exhaustive search. 

To demonstrate the gain expected from  backtracking, we illustrate in Figure~\ref{backvsExh} the search spaces 
that we explore. This is the case where, for the graph of Figure~\ref{noSolution}, an SLA constraint states that the service response time must be lower than $13$ms. 

\begin{figure}[htbp]
\centering
\fbox{
\includegraphics[width=0.9\linewidth,height=1.7in]{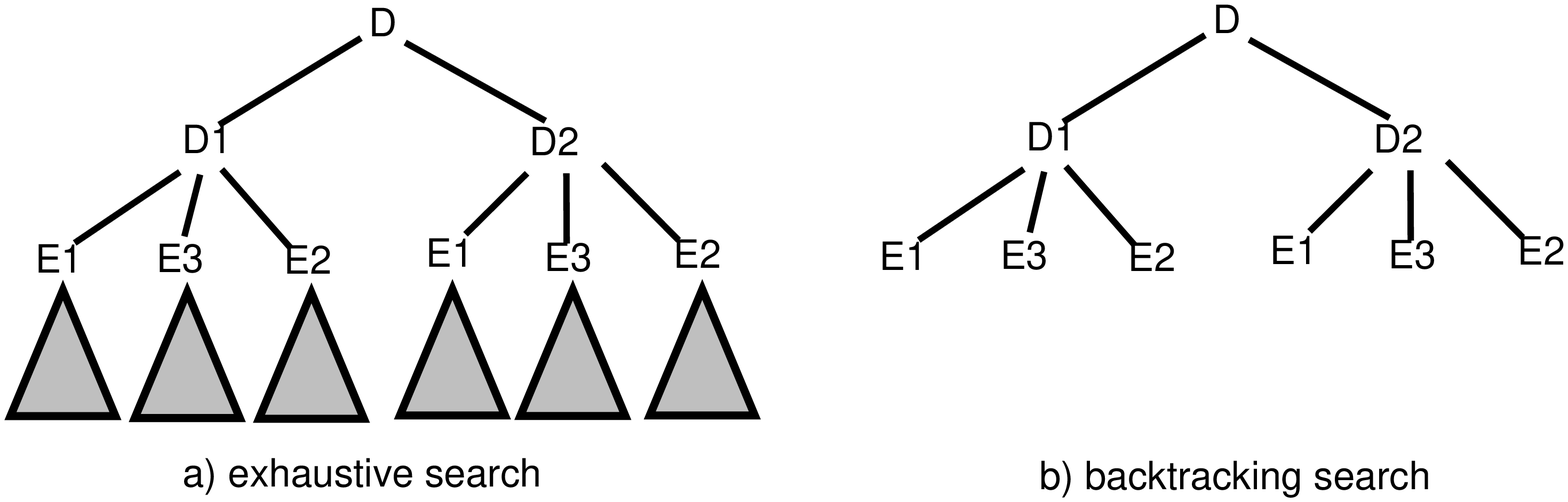}}
\caption{Exhaustive search (in a) vs backtracking search (in b). The dark sub-trees correspond to 
assignments made to $A, B, F, G, H$. With backtracking, these assignments will not be explored.}\label{backvsExh}
\end{figure}

For applying the backtracking technique, we will now discuss two points. Firstly, we present 
the basic principles that we use for ordering variables. Then we discuss the implementations 
of the principles in a backtracking algorithm for service selection.

\subsubsection{Ordering principles}

The ordering of variables is important in backtracking. The literature~\cite{Baker95intelligentbacktracking} proposes 
multiple static orderings. We propose to consider two of the most popular: the min-domain and the 
max-degree ordering. In the  min-domain ordering, the abstract operations whose concrete set of 
operations are smaller must be considered first in the assignments. In the max-degree ordering, the abstract operations 
that are the most connected to the other operations will be considered first. We map these orderings in the 
resolution of the service selection 
problem by the mean of two principles that we introduce further. Let us first consider the following definitions.  

\begin{definition}[Correct partial evaluation]
Given a decomposable HSG $H$ whose abstract operations are bound to concrete one, we define a correct partial evaluation 
of QoS as the vector of QoS values (energy consumption, response time) of a decomposable subgraph of $H$.  
\end{definition}

For instance, given the operations graph of Figure~\ref{noSolution}, if $D$ is assigned to $D_1$ and $E$ to $E_1$, then a 
correct partial evaluation is the vector $(22, 1.6)$ for the response time and energy consumption of 
the subgraph $(g_1, g_2)$. We only compute partial evaluations on decomposable graphs. Let us recall that 
such graph are obtained by composing the regular patterns that we consider in the semantics of the operations 
graph. The objective in making partial evaluations is twice: (1) the partial evaluations bounds are compared 
to  the SLAs bounds to check if there is a violation; (2) these bounds are compared to the local optimum found for 
the service selection problem to see if it is not already dominated. 

It is important to notice that in the comparisons, the partial evaluations bounds must in some cases be weighed. 
In  Figure~\ref{noSolution}, let us consider the QoS vector (SRT, EC) issued from the reduction of $(g_3, g_4)$. 
Since this subgraph is included in the xor subgraph $(g_5, g_6)$, we cannot directly compare the response time 
of this vector to the SLA bound on response time. For taking into account the semantics, we need to multiply this value 
with the probability for a request to be routed towards $g_3$. In our prior work~\cite{mgc2012}, for such situations, 
we introduced the notion of {\it reachability probabilities}. In a simplified manner, for each subgraph, this gives 
the probability for a request to be routed to it. We will use these probabilities for weighting the comparison of QoS 
vectors with SLAs bounds.

The correct partial evaluations capture a facet of our backtracking algorithms: we will regularly evaluate sub-assignments 
of concrete services to abstract ones such as to detect whether or not we must continue in this exploration 
path. As stated before, the order of assignments of concrete services to abstract ones can have a great 
influence in the run. We introduce below the partial evaluation precision for characterizing the possible orderings. 

\begin{definition}[Partial Evaluation Precision]
Given a decomposable HSG $H$ and an ordered list of abstract operations $L = [ u_1, \dots, u_k ]$, 
we define the precision of the partial evaluation of $L$ as the difference $\displaystyle pep(L) = |L| - neo(L)$. 
In this formula, $neo(L)$ is the maximal number of $L$' operations from which one can compute a correct 
partial evaluation of QoS from $H$. 
\end{definition}

This definition relies on the fact that any partial assignment will not lead to a lower bound 
on SRT and EC that includes all abstract operations. At the beginning of the backtracking algorithm, we must define 
an ordered list $L$ of abstract operations. This list is such that $L[1]$ is the first abstract operation that will be 
assigned, then follows  $L[2]$ and so one. At a moment in the algorithm, we could have assigned a concrete 
operation to $L[1], L[2], \dots, L[i]$; but it might happen that the assigned operations do not 
constitute a decomposable graph (See Section~\ref{evaluationAlgorithm} for decomposable graph). 
In these cases, a correct partial evaluation will be obtained only 
from a subset of the assignments ($neo(L)$). In Figure~\ref{noSolution} for instance if $L = [D, A]$ 
then $pep(L) = 2$. If $L = [D, E]$, then $pep(L) = 0$. Indeed, $[D, A]$ do not {\it shape} a decomposable 
graph. 
We can generalize the notion of precision.
Let us assume that $L[1..i]$ defines the sublist having the first $i$ elements of $L$.  

\begin{definition}[Total Evaluation Precision]
Given a decomposable HSG $H$ and an ordered list $L = [ u_1, \dots, u_k ]$ of its abstract operations, 
we define the precision of the total evaluation of $L$ as $\displaystyle \sigma(L) = \sum_{i = 2}^{|L|-1} pep(L[1..i])$.
\end{definition}
$\sigma(L)$ captures the distance between two numbers of operations: those to 
which a concrete service is assigned and those from which a correct evaluation can be made. The following 
result is straightforward.

\begin{property}
Let us consider a decomposable HSG for which all abstract operations are assigned to 
concrete ones. Let us also consider an ordered list $L = [ u_1, \dots, u_k ]$ of its abstract operations. Then,  
$pep(L[1..i]) \geq 0$,  $pep(L[1..k]) = 0$ and  $\displaystyle \sigma(L) \geq 0$. 
\end{property}

Based on these definitions, we can then define the first principle that we will use for 
ordering abstract operations.

\begin{principle}[Partial Evaluation of QoS First (PEQF)]
Let us consider a decomposable HSG $H$ for which all abstract operations are assigned to 
concrete ones. The generated ordering list $L$ for $H$ must minimize the precision of the total  
evaluation of $L$.
\end{principle}

With this principle, our objective is to maintain, each time during the search, an updated correct 
partial evaluation that we can use for checking whether or not SLAs are violated. As one can remark, this 
is not the case with large values of $\sigma(L)$. 
In Figure~\ref{noSolution} for instance, with the ordering $L = [B, A, E, D]$, we have $\displaystyle \sigma(L) = 5$ ( 
$[B, A, E]$ do not describe a decomposable graph). In this case, the backtracking will 
not improve the exhaustive search. This is because we must wait for the assignment of a concrete operations to 
each abstract one for expecting a QoS evaluation. In choosing however 
the ordering $E, D, A, B$ we have $\displaystyle \sigma(L) = 0$. As one can observe, we can 
quickly have here a correct partial evaluation that can then be used for checking SLAs violation. 

Regarding the implementation of the PEQF principle, it is important to consider the following 
result. 

\begin{property}
We can find a decomposable HSG $H'$ for which there exist two ordering lists $L_1$ and $L_2$ of 
abstract operations for which $\sigma(L_1) = \sigma(L_2) = 0$. 
\end{property}

This is the case in Figure~\ref{noSolution} with the lists:  
$L_1 = [D, E, A, B]$ and  $L_2 = [D, E, B, A]$. The question in these settings is to choose among the 
two lists.  We adopt for this following principle.

\begin{principle}[Min Domain First (MDF)]
The ordering of abstract operations must consider in priority, the correct partial evaluations with 
the shortest number of concrete operations. A random ordering must be adopted if we have multiple 
options. 
\end{principle}

This principle is inspired from the min-domain heuristic in constraint satisfaction. Let us indeed assume that 
$B$ has less concrete operations than $A$ (i.e. $|C_o(B)| < |C_o(A)|$). Then, the ordering to choose in Figure~\ref{noSolution} 
is $L_2 = [D, E, B, A]$.
The objective is to detect quickly invalid assignments by considering small domains before. One can 
criticize the choice of min-domain because in CSP resolution, the max-degree heuristic  
also performs well for detecting invalid assignments. However, let us observe that the idea of 
max-degree (to start with the most connected variable) is partially included in the first principle (PEQF). 
Indeed, we will see that for finding quickly decomposable graphs, we must consider nested operations in priority. 

To summarize, we modelled our ordering goals within principles, below we consider their implementation for deriving 
backtracking algorithms. 

\subsubsection{Implementation of the PEQF principle}

For implementing the PEQF principle, we propose to use the ordering of abstract operations suggested by 
the reduction order of the evaluation algorithm. For instance, in Figure~\ref{noSolution}, from the 
reduction order $(g_1, g_2); (g_1, B), (A, g_1)$, we deduce the possible ordering $E, D, B, A$. 
What is challenging is to derive systematically such an ordering. For this we will  
introduce two data structures. We will also manipulate the deepness concept, introduced in 
prior work~\cite{GoldmanNgoko}. Below, we recall its definition. 
\begin{definition}[Deepness~\cite{GoldmanNgoko}]
Given an operation graph $G_o$, let us suppose that for a node $u$ (operation or connector),  
we have $n$ paths $Pt_1, \dots Pt_n$ leading to it. In each path $Pt_i$, we have $\alpha_i$ 
split connectors and $\beta_i$ join connectors. The deepness of $u$ is defined 
as $deep(u) = \underset{1 \leq i \leq n}{\max} \{\alpha_i - \beta_i\}$.
\end{definition}
For example in Figure~\ref{dataStructure}, $deep(A) = deep(g_1) = 0$, $deep(B) = deep(E) = 1$. 
The first data structure that we consider for the implementation of the PEQF principle is the  
 {\it nested list for subgraphs' nodes (NeLS)}. For split/join subgraphs, the list gives the 
operation nodes whose deepness are equal to the one of the split node plus $1$. In Figure~\ref{noSolution}, 
the NeLS will have an entry $(g_1, g_2)$ pointing towards a list with the operations $D$ and $E$. 
This is because we have a unique split/join graph that comprises these
operations. For the operations graph of Figure~\ref{dataStructure}, the NeLS has two 
entries. The first entry ($(g_3, g_4)$) points towards a list with $C, D$. The second entry ($(g_1, g_2)$) 
points towards $B, E$. Let us remark that we did not include $C$ here because
$deep(C) = deep(g_1)+2$. Given a NeLS $h_s$, we will use the term $hs(x,y)$ for referring 
to the list that the entry $(x,y)$ points to. For instance, if $h_s$ is the name of the NeLS of 
Figure~\ref{dataStructure}, then  $h_s(g_3, g_4)$ points towards $C$ and $D$. 

\begin{figure}[htbp]
\centering
\fbox{
\includegraphics[width=0.8\linewidth,height=2.5in]{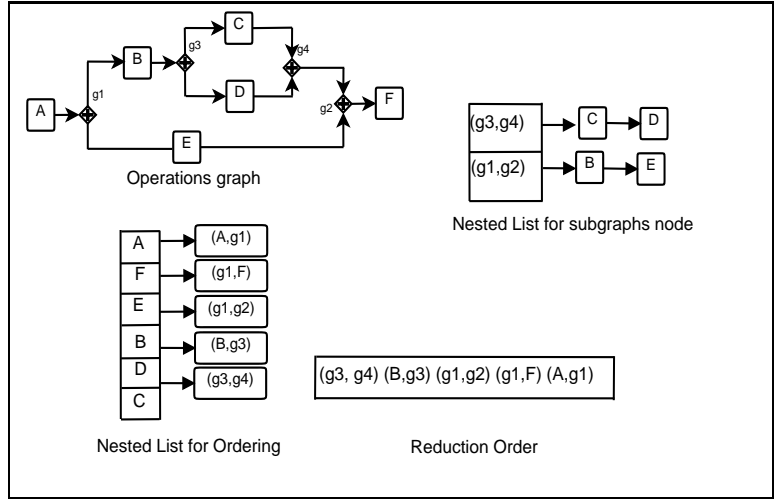}}
\caption{Data structures for the PEQF principle}\label{dataStructure}
\end{figure}

The second data structure is the {\it nested list for operations' ordering (NeLO)}. The main entries of a NeLO consist 
of a list of ordered operations. Each entry points towards a list of subgraphs, defined with their root 
and leafs. The idea is that once a value is assigned to the abstract operation in one entry, the pointed 
subgraphs can be reduced. 

In our algorithms, while the NeLO will be used for storing the ordering of abstract operations, we 
will use the NeLS for the generation of this order. In particular, {\it the NeLS will serve us for 
detecting when to assign a value to an abstract operation that is not included in the reduction order}. 
Figure~\ref{dataStructure}  shows a NeLO related to an operations graph. This NeLO describes the ordering in 
which abstract operations must be considered in backtracking assignments. The first 
operation to which a concrete one must be assigned is $C$. This entry does not point towards any list. 
Therefore, after this assignment, no reduction is done. Then, we must consider $D$. The assignment of 
a concrete operation to $D$ implies that we can reduce the subgraph $(g_3, g_4)$. 
Then, we must continue with $B, E,$ and so on. 

Let us notice that it is easy to redefine the notion of {\it total evaluation precision} for 
computing the value of $\sigma(h_o)$ given the NeLO $h_o$. This is because the entries of a NeLO 
constitute an ordered list of abstract operations. 

With the defined data structures, we will now state how we implement the 
PEQF principle. We view the implementation of the PEQF principle as the computation of a NeLO for which the 
total precision is minimal. The computation of this NeLO is done from an operations graph and a NeLS. 
The process is the following. Firstly, from the operations graph, we generate a reduction order and a NeLS. 
Then, we pick the top element of the reduction order. This element corresponds to a pair $(x,y)$ defining the 
frontiers of a subgraph. We process this element for generating new entries in the NeLO to build. The rules of 
this processing are given in Figure~\ref{Cases}. Once $(x,y)$ is processed, we consider the next element of the 
reduction order. We continue in the same way until processing the last element 
of the reduction order. In Figure~\ref{dataStructureExample}, we illustrate the application of 
this process.

\begin{figure}[!htbp] 
\begin{center}
\begin{tabular}{|p{13cm}|}
\hline 
\small [We have a pair $(x,y)$ in the reduction order and a NeLS denoted $h_s$]\\
\small case \#1 [$x$ and y are operations]: we create two new entries referring to $x$ and $y$ in the NeLO. 
We chain the last entry with a list pointing towards $(x,y)$; this is for stating that this subgraph 
can be reduced after assigning a value to $x$ and $y$. We then remove $x$ and $y$ from all lists of $h_s$.\\
\small case \#2 [$x$ is an operation and $y$ is a split connector]: we create a unique entry $x$ and chain it with a list 
towards $(x,y)$. We remove $x$ from all lists of the the $h_s$. \\
\small case \#3 [$x$ is a split connector and $y$ is an operation]: we create a unique entry $y$ and chain it towards 
$(x,y)$. We remove $y$ from all lists of the the $h_s$.  \\
\small case \#4 [$x$ is a split connector and $y$ is a join connector]: If the list $h_s(x,y)$ is not null, then we 
create an entry in the NeLO for all elements of $h_s(x,y)$. We chain the list of the last element of the NeLO 
with $(x,y)$. Then, we delete $h_s(x,y)$.  \\
\hline
\end{tabular}
\caption{Processing of an element $(x,y)$ for the generation of the operations hash table.}
\label{Cases}
\end{center}
\end{figure}

\begin{figure}[htbp]
\centering
\fbox{
\includegraphics[width=1\linewidth,height=2.6in]{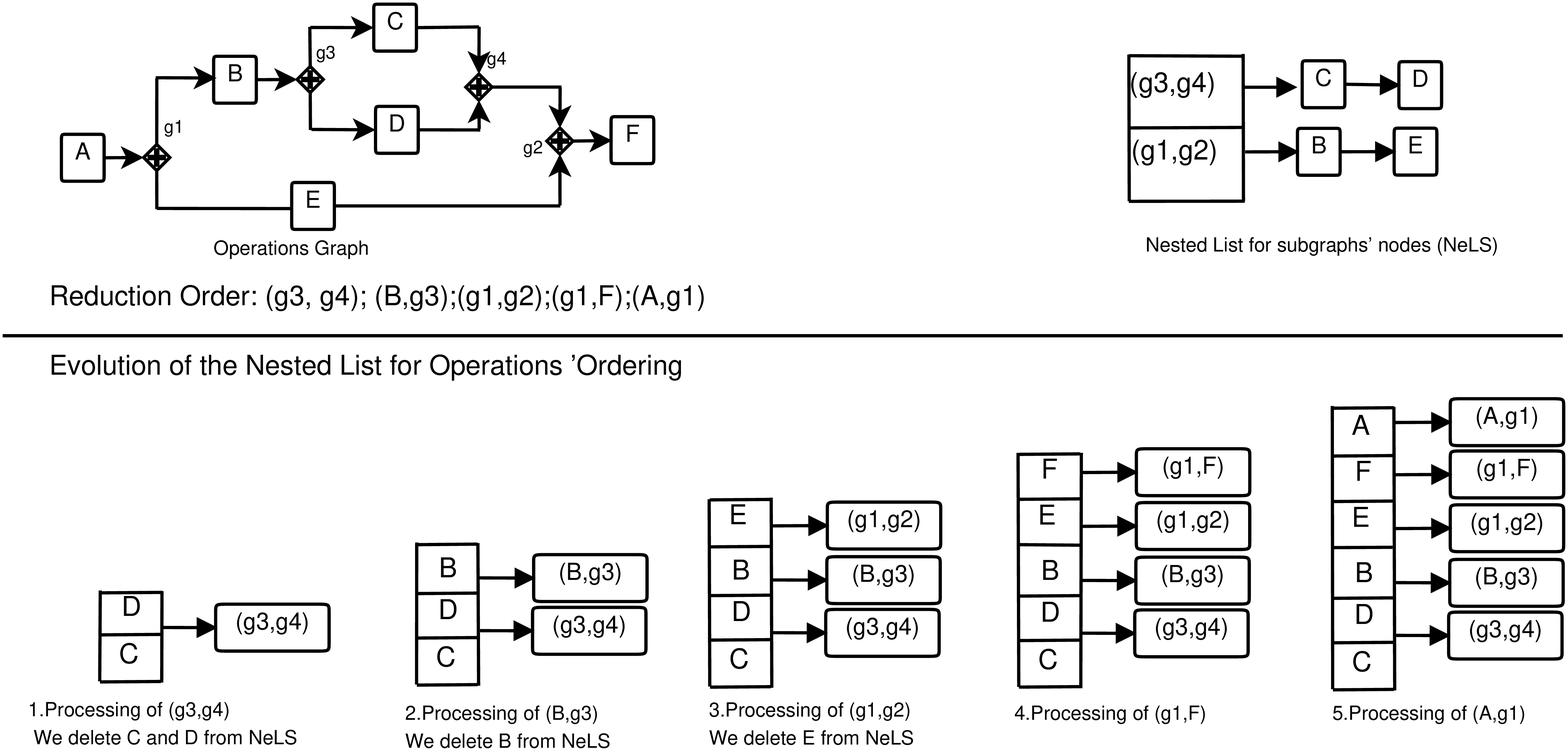}}
\caption{Computation of the ordering}\label{dataStructureExample}
\end{figure}

One can understand the role of the NeLS in this computation as follows. The objective in the NeLO 
is to have a list of abstract operations pointing towards reductions to be done. For obtaining  
 partial evaluations quickly, we generate the NeLO based on the reduction order. However, the representation of 
this order might not include some operations. For instance, in Figure~\ref{dataStructureExample}, $C$ and $D$ 
do not appear in the reduction order. In the NeLO generation, we find the missing operations from the NeLS. 
In particular, when we have a subgraph reduction, we first explore the NeLS (see case \#3) for including 
the subgraph' operations in the NeLO. As one can remark in  Figure~\ref{dataStructureExample}, $C$ and $D$ are 
referred to in the NeLS. 

There are two important observations. The first one is that given an 
element $(x,y)$ of the reduction order, in Figure~\ref{Cases}, we consider all cases regarding 
the type of $x$ and $y$~\footnote{The cases are described in Section~\ref{evaluationAlgorithm}}. 
The second observation is that the described generation process has a polynomial runtime in the number of nodes 
of the operations graph. More precisely, we have the following result.

\begin{lemma}
Giving an operations graph, a reduction order and a NeLS, the NeLO generation can be 
done in $O(n^2)$ where $n$ is the number of nodes of the operations graph. 
\end{lemma}

\begin{proof}
Firstly, let us observe that the generation process that we described loops on the number of 
elements of the reduction order. If we have $n$ nodes in the operation graph, then our prior 
work~\cite{GoldmanNgoko} guarantees that the number of elements of the reduction order is in $O(n)$. 
In the treatment of an element of the reduction order, we have the cases listed in Figure~\ref{Cases}. 
These cases are dominated by two instructions: the creation of NeLO entries and the deletion of NeLS elements. 
Since the number of elements of NeLO and NeLS are in $O(n)$, these two instructions are in $O(n)$. 
We have the proof in considering that we loop on elements of the reduction order. 
\end{proof}

Regarding the PEQF principle, the quality of the process proposed for NeLO computations 
can be perceived throughout the following result. 

\begin{lemma}
Let us consider an operations graph with $n$ abstract operations. If the maximal outgoing degree of 
a split connector in the graph is $2$ then, the generated NeLO 
$h_o$ is such that $\displaystyle \sigma(h_o) \leq \frac{n}{2}$. If the graph is a sequence, then 
$\displaystyle \sigma(h_o) = 0$. 
\end{lemma}

\begin{proof}

We obtain the result from an analysis of the process of Figure~\ref{Cases}. Firstly, let us assume that 
the graph is a sequence. 
The first element $(x,y)$ of the reduction order in this case refers to two operations. 
According to the process of Figure~\ref{Cases}, this will generate two entries in the NeLO such that the 
last entry points towards the reduction $(x,y)$. Consequently, $pep(h_o[1..2]) = 0$. According to 
the reduction order algorithm~\cite{GoldmanNgoko}, the second element $(x',x)$ of the reduction order will refer to two 
operations $x'$ and $x$ (already in the NeLO). Therefore, we will put an operation in the NeLO entry with a reduction to operate. 
This implies that $pep(h_o[1..2])+pep(h_o[1..3]) =  0$. The third element of the reduction order will have the form $(x",x')$. 
Consequently $pep(h_o[1..2])+pep(h_o[1..3])+ pep(h_o[1..4])=  0$. Generalizing, we will have for the resulting NeLO $h_o$, 
$\sigma(h_o) = 0$. 

Let us now assume that we have split connectors in the operation graph. Then, for the first element $(x,y)$ of the reduction order, we have two cases. Either we have two operations or we have a split connector ($x$) and a join ($y$). 
There are no other possibilities from the reduction algorithm. In the first case, we can easily guarantee from what 
precedes that $pep(h_o[1..2]) = 0$. In the second case, the process of Figure~\ref{Cases} states that we will 
add two operations in the NeLO entry such that the last operation will point towards the reduction $(x,y)$. Consequently, 
 $pep(h_o[1..2]) = 0$. For the processing of the next element $(x', y')$ of the reduction order, we can have multiple cases. Either 
$x'$ or $y'$ is an operation, or, they correspond to connectors. In the former case, we can ensure that $pep(h_o[1..3]) = 0$; 
in the latter case, we can ensure that $pep(h_o[1..3])  \leq 1$ and $pep(h_o[1..4]) \leq 1$. Generalizing, 
$\displaystyle \sigma(h_o) \leq \frac{n}{2}$

\end{proof}

An interesting question is the one of knowing whether or not better total precisions can be expected. 
The answer is no. For sequences graphs, the optimality of the result is guaranteed. For arbitrary structures, we have a lower 
bound. Indeed, let us consider a sequence of subgraphs of two elements. On such graphs, it is impossible to build an ordering 
$h_o$ such that $\sigma(h_o) < \frac{n}{2}-1$. This comes from the fact that for reducing an internal subgraph, 
we must assign a concrete operation to each abstract one. When assigning a value to the first concrete operation, an evaluation is not possible. It is only when making 
an assignment to the second operation that we can evaluate.  Consequently, we can consider that the proposed implementation of the 
PEQF principle is optimal. In the following, we state how to implement the MDF principle.

\subsubsection{Implementation of the MDF principle}

The objective in the MDF principle is to start the assignments with abstract operations whose set of 
concrete ones are small. The applicability of this principle however can be in conflict with the implementation 
of the NeLO. It is the case in Figure~\ref{dataStructureExample} 
if $|Co(F)| > |Co(A)|$. Indeed, the evaluation of $(A, g_1)$ will concern a domain that is smaller than the 
evaluation of $(g_1, F)$. However the latter evaluation is done before in the NeLO. Therefore, how can we 
conciliate the two principles? For this, we consider the following result. 

\begin{lemma}[Free permutation of operations and subgraphs]\label{freePermutation}
Let us consider a decomposable HSG $H$; \\ a) let us assume that in  the operations graph, we have a sequence $(x, y)$  where $x$ and $y$ are either 
operations or decomposable subgraphs. The HSG $H'$ in which we reversed the subgraphs $x$ and 
$y$ (so we have the sequence $(y, x)$) has the same mean response time and energy consumption as $H$; \\ b) let us assume that 
we have a split connector $g$ in $H$. The $H'$ in which we reversed two branches of the split connector 
has the same mean response time and energy consumption as $H$. 
\end{lemma}

\begin{proof}
The results come from the commutativity of computations in the QoS aggregation rules (see Table~\ref{tabAggRules}). 
The response time of the graph $(x, y)$ will be $S(x)+S(y) = S(y) + S(x)$. The energy consumption will be 
$E(x)+E(y) = E(y) + E(x)$. Since $S(y) + S(x)$ and $E(y) + E(x)$ are the response time and energy consumption of 
$(y,x)$, we have the proof in the case of sequences. In the case of split connectors, we can establish the 
proof in the same way. For instance, given an elementary Fork with the operations $x, y$ its response time 
is $\max\{x, y\} = \max\{y, x\}$.
\end{proof}

The interest in this result is that it suggests a solution for applying the MDF principle 
without violating the PEQF principle. The idea is that the generation of the NeLO is based on a particular 
exploration of the operations node. Before this generation, one can use the reverse instructions of Lemma~\ref{freePermutation} 
for obtaining a topology of the operations graph, adjusted to  partial evaluations for  
small domains in priority.  
As a simple illustrative example, let us reconsider the operations graph of Figure~\ref{dataStructureExample}. 
In the case where $|Co(F)| > |Co(A)|$ we can switch the nodes $A$ and 
$F$. As a result, the reduction order will be $(g_3,g_4); (B, g_3); (g_1,g_2); (g_1, A); (F,g_1)$. In it, 
the evaluation $(g_1, A)$ is now {\it before} $(F,g_1)$.

For the implementation of the MDF principle, we propose to make a topological sorting of the initial operations 
graph that is done based on an extended NeLS (denoted NeLS+). The extensions concern the following points: a) we distinguish between 
branches of the elements of each split/join subgraph; b) we include split/join subgraphs in the list towards which 
a NeLS entry points; c) We assume that the root and leaf nodes of the operations graph are branches of a 
virtual graph $(g_0, g'0)$.

\begin{algorithm}[H]                    
\begin{algorithmic}[1]
\scriptsize
\Function{Main}{}
	\State Generate the reduction order and store it in $ORD$;
	\State Build a NeLS+ $hs^+$;
	\State Topological Sorting of $H$ according to $hs^+$. The result is $H'$;
	\State Build a NeLS $hs$ for $H'$;
	\State Generation of a NeLO $h_o$ from $H'$ and $hs$;
        \State $OptPenalty = +\infty$;  $index = 0$;
	\State Create an uninitialized array $f$ of values for abstract operations;
	\State Call backtrack($f$, $H'$, $h_o$, $OptPenalty$, $index$);
	\State Return the best assignment and $OptPenalty$;
\EndFunction

\Function{backtrack}{$f$, $H'$, $h_o$, $index$}
\If{$index = |f|$}
	\State Compute $E(f)$ and $S(f)$ from the evaluation algorithm with $H$ and $Q$;
	\If{ $S(f) \leq MaxS$ and $E(f) \leq MaxE$ }
		\If{$\lambda.S(f) + (1-\lambda).E(f) < OptPenalty$}
			\State Save $f$ as the best assignment;
			\State $OptPenalty = \lambda.S(f) + (1-\lambda).E(f)$;
		\EndIf
	\EndIf
	\State Return;
\EndIf
\For{ all concrete operations $u \in Co(h_o[index])$}
\State $f[index] =$ $u$;
\If{ $h_o[index]$ points towards some reductions}
\State  Update $E(f)$ and $S(f)$ by making the reductions;
\State Get the reachabilities probabilities $pa$ of the last reduction;
\If{ $pa.S(f) \leq MaxS$ and $pa.E(f) \leq MaxE$ }
\If{ $\lambda.pa.S(f) + (1-\lambda)pa.E(f) < OptPenalty$ }
\State Call backtrack($f$, $H$, $h_o$, $OptPenalty$,  $index+1$);
\EndIf
\EndIf
\Else
\State Call backtrack($f$, $H$, $h_o$, $OptPenalty$,  $index+1$);
\EndIf
\EndFor
\EndFunction

\normalsize
\end{algorithmic}
\caption{\scriptsize SS-b-PM (Backtracking search for service selection with the PEQF and MDF principle). \\ {\bf INPUT:} a HSG $H$ and a QoS matrix $Q$ giving the energy consumption and service response time of each concrete operation;  \\ {\bf OUTPUT:} An assignment of concrete operations to abstract ones }
\label{alg:Backtracking}  
\end{algorithm}

Based on the lemma~\ref{freePermutation}, we will consider two main instructions: the {\it branch sorting} and the 
{\it sequence sorting}. In branch sorting, the objective is to revert the branches of a split/join subgraph such as to 
ensure that when computing the reduction order, in  the first place, we will consider the branch with the smallest domain. 
Such an instruction is meaningful because the computation of the reduction order occurs through a 
depth first search algorithm in which the branches are explored according to a number assigned on them. 
In our current implementation, the branch with the greatest number is explored in priority.

The  sequence sorting considers a sequence of operations and subgraphs and revert them such as to ensure 
that the operations or subgraphs with smallest domain will be considered in priority. In the computation of the 
reduction ordering, let us observe that the last operations of each sequences are considered in priority. Therefore, 
the reverting must ensure that these operations have the smallest domain size. In Figure~\ref{transClosureExample} 
for example, the sequence sorting will reverse in the entry  $(g_0, g'_0)$ the nodes $A$ and $F$. 

The topological sorting is done as follows. Initially, we sort the entries of the NeLS+ based on their deepness. 
The idea is to have the deeper subgraph at the top and the less deep one at the end. 
Then, we compute the total domain size to which correspond 
each entry. The domain size of an entry is the sum of domain size of the lists of elements to which 
it gives access. Let us remark that because of the initial sorting, we can simply evaluate the domain size starting from the 
top entry to the last one. In Figure~\ref{transClosureExample} for instance the domain size of $g_3$ is {\it dom-size($g_3,g_4$)} = $|Co(C)| + |Co(D)|$. The domain size of $(g_1, g_2)$ is $|Co(E)| + |Co(B)| + $ {\it dom-size($g_3,g_4$)}. 
Once we have the domain size of each entry, we make branches sorting first and sequence sorting next. 

\begin{figure}[htbp]
\centering
\fbox{
\includegraphics[width=0.75\linewidth,height=1.4in]{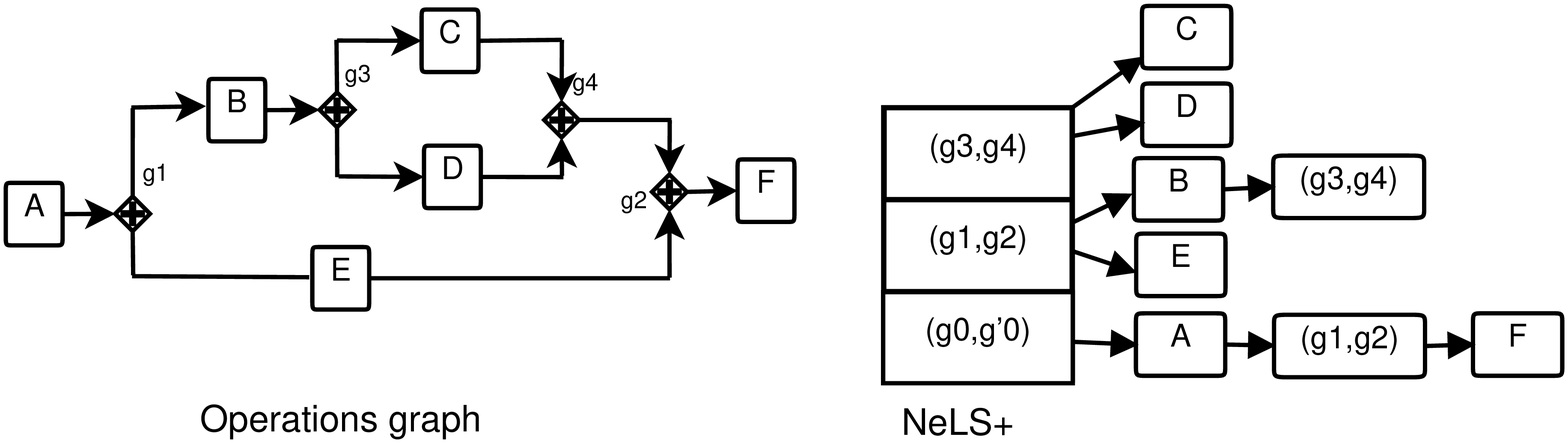}}
\caption{Example of NeLS+}\label{transClosureExample}
\end{figure}

Now that we stated how we implement the principles, in Algorithm~\ref{alg:Backtracking}, we give  
our general backtracking scheme. 
The proposed scheme is based on the following notations:
\begin{itemize}
\item  We assume that $h_o[i]$ refers to the variable of $h_o$  in the $i^{th}$ entry.
\item For defining assignments, we use the array $f$. $f[i]$ will comprise the concrete service 
associated with $h_o[i]$.
\item The over-defined notion $Co(h_o[index])$ refers to the set of concrete services that can be assigned to $h_o[index]$.
\end{itemize}

We will refer to this global algorithm as {\it SS-b-PM}. We will also consider its variant {\it SS-b-P} 
where we do not apply the MDF principle.

\section{Experimental evaluation} \label{ExperimentalEvaluation}

 Throughout the experimental evaluations, our objectives were the following:
\begin{enumerate}
\item to demonstrate that the backtracking based heuristics outperforms the exhaustive search;
\item to demonstrate that the backtracking based heuristics outperforms integer linear programming;
\item to compare the different heuristics.
\end{enumerate}

\subsection{Backtracking versus exhaustive search}

For these experiments, we used $4$ types of operations graphs, each based on 
reference BPMN processes~\cite{Omg,Freund}. The structure of the processes is given 
in Figure~\ref{figProcess}.

From each processes, we created $300$ service selection problems. 
Depending on the SLAs, we ranged instances in three classes of $100$ instances: simple, medium, hard. 
We chose these names because our experiments globally showed  that in increasing the size of the bounds, 
$MaxS$ and $MaxE$, we increase the runtime of all algorithms. Intuitively, the reason is that with 
big bounds, there are more candidate solutions.  
Depending on the number of concrete implementations per abstract operations, 
we ranged our instances in $5$ classes of $60$ instances each. The setting of instances 
is resumed in Table~\ref{instanceSetting}. 
For each instance, we randomly draw the service response time of operation between 
$1,\dots, 1500$. Given a service response time $S$, we deduced the energy consumption from 
the formula $E = P.S$, where $P$ is a power consumption value randomly drawn  between $100,\dots,150$.

\begin{figure}[ht]
\centering
\fbox{
\subfloat[Shipment]{
\includegraphics[scale=0.30]{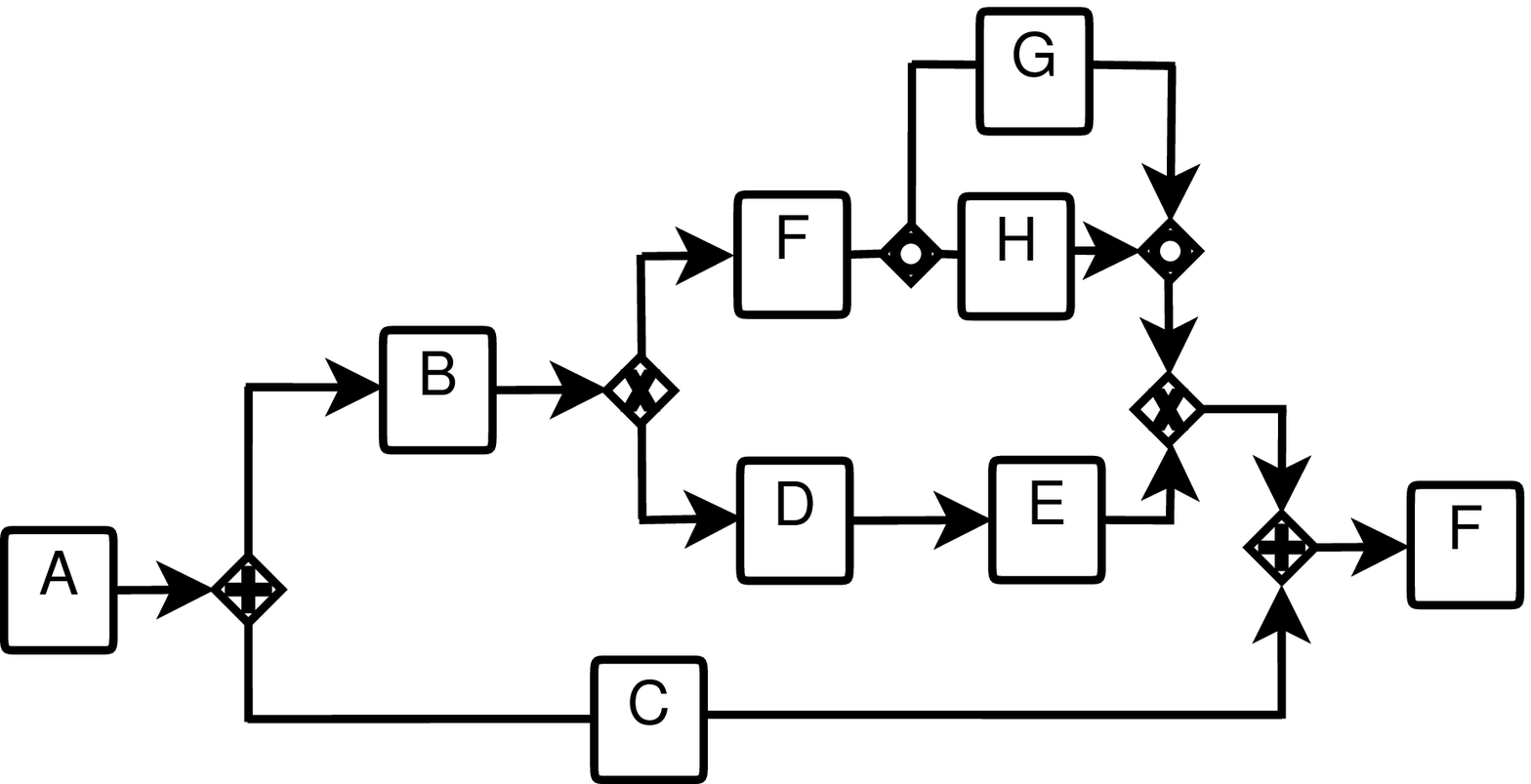} 
}
\subfloat[Procurement]{
\includegraphics[scale=0.30]{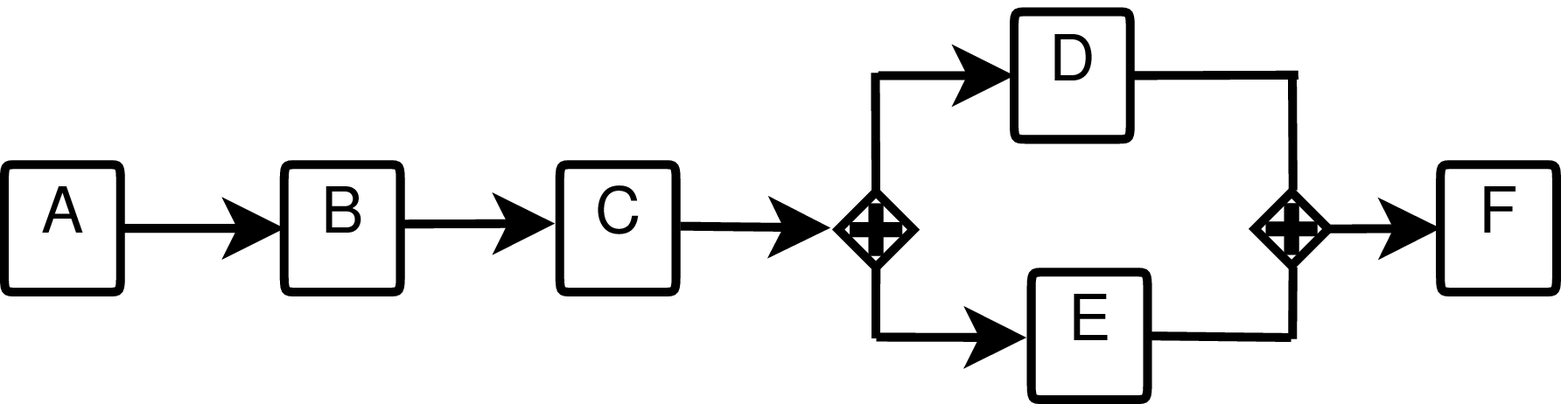}
}
}
\fbox{
\subfloat[Disbursement]{
\includegraphics[scale=0.30]{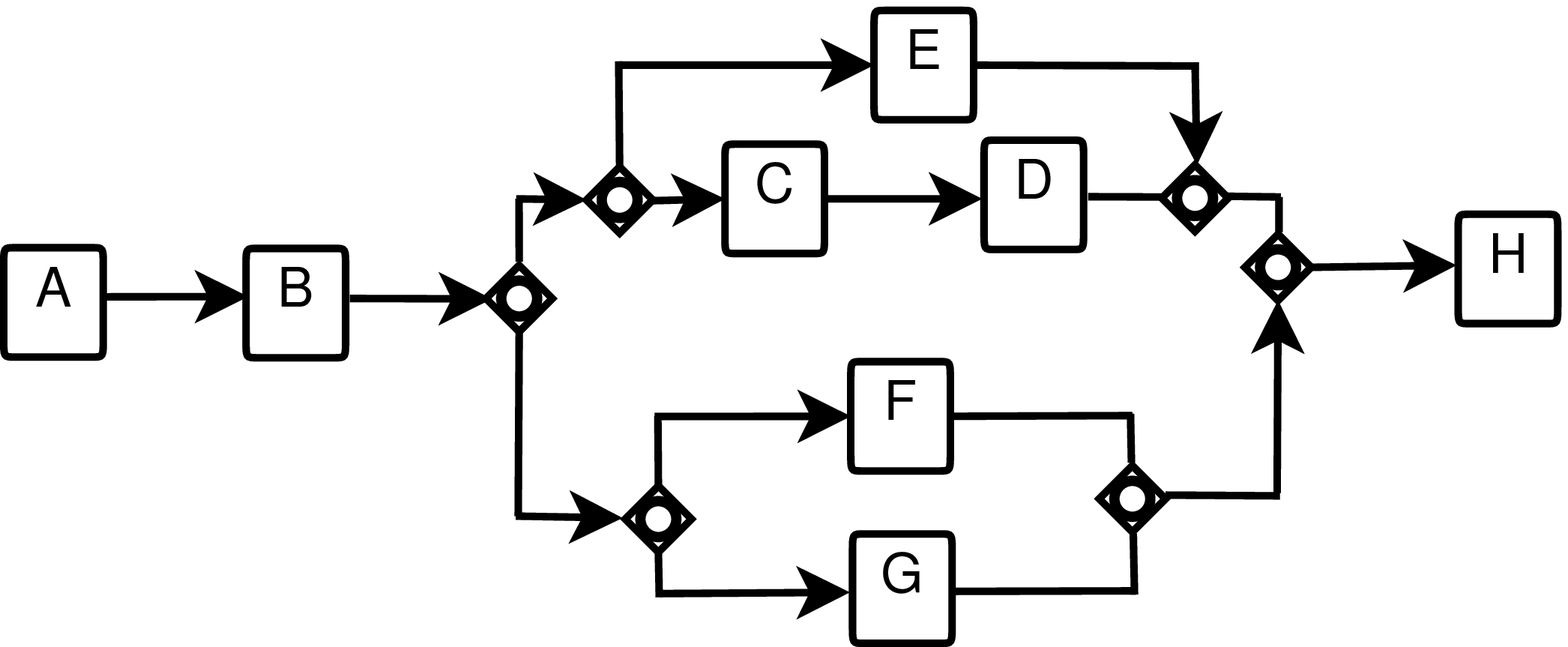}
}
\subfloat[Meal Options]{
\includegraphics[scale=0.30]{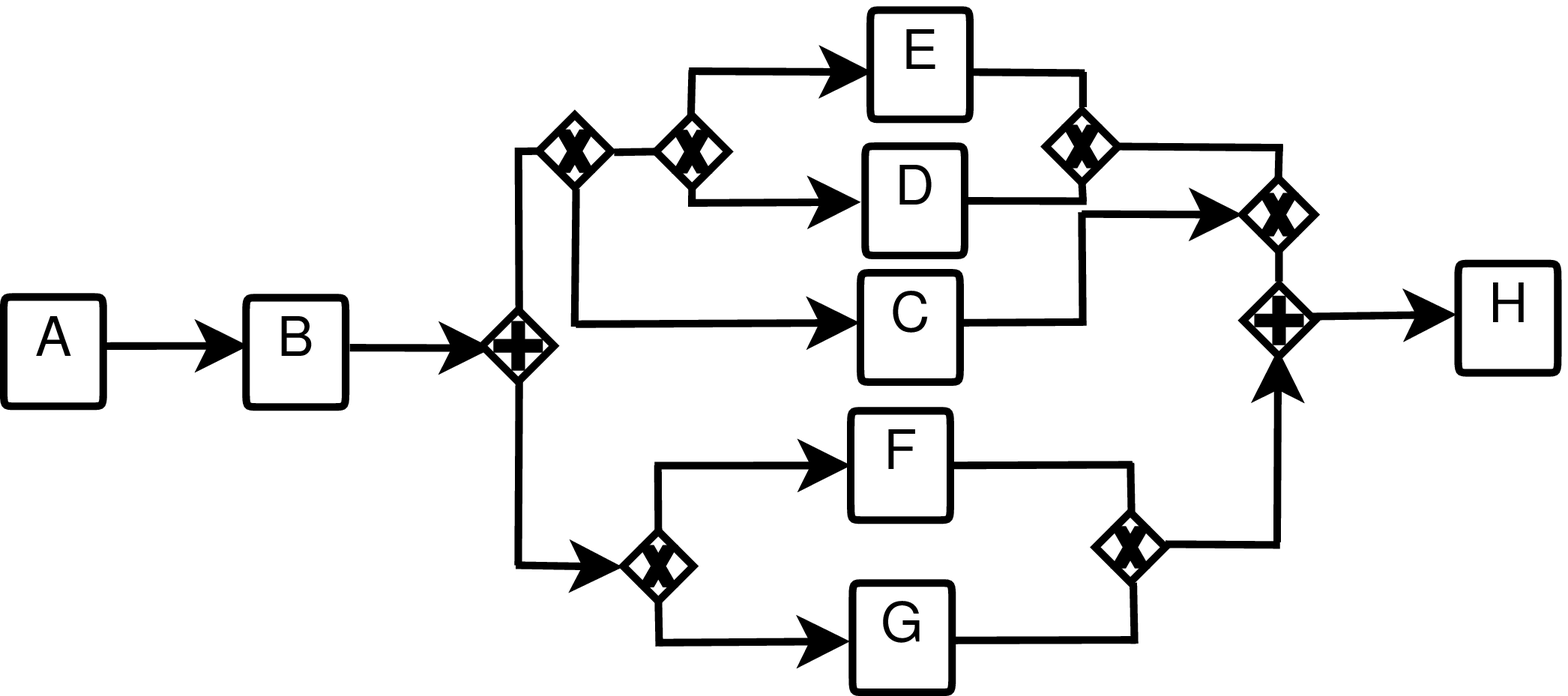}
}
}
\caption{Process examples}
\label{figProcess}
\end{figure}

\begin{table}[htbp]
\centering 
\begin{tabular}{c|c|c}
\hline
{\bf Type of SLAs }  & {\bf (MaxS; MaxE)} & {\bf Domain sizes}  \\\hline
  \textit{Simple}   & {$(3500; 7000)$}   & \multirow{3}{*}{$4$, $6$, $8$, $10$, $12$} \\
  \textit{Medium}   & {$(3000; 5500)$}   \\
  \textit{Hard}   & {$(2700; 4000)$}   \\\hline
\end{tabular}
\caption{Instance settings}
\label{instanceSetting}
\end{table}

\begin{figure}[ht]
\centering
\subfloat[Shipment]{
\includegraphics[width=0.45\linewidth,height=1.6in]{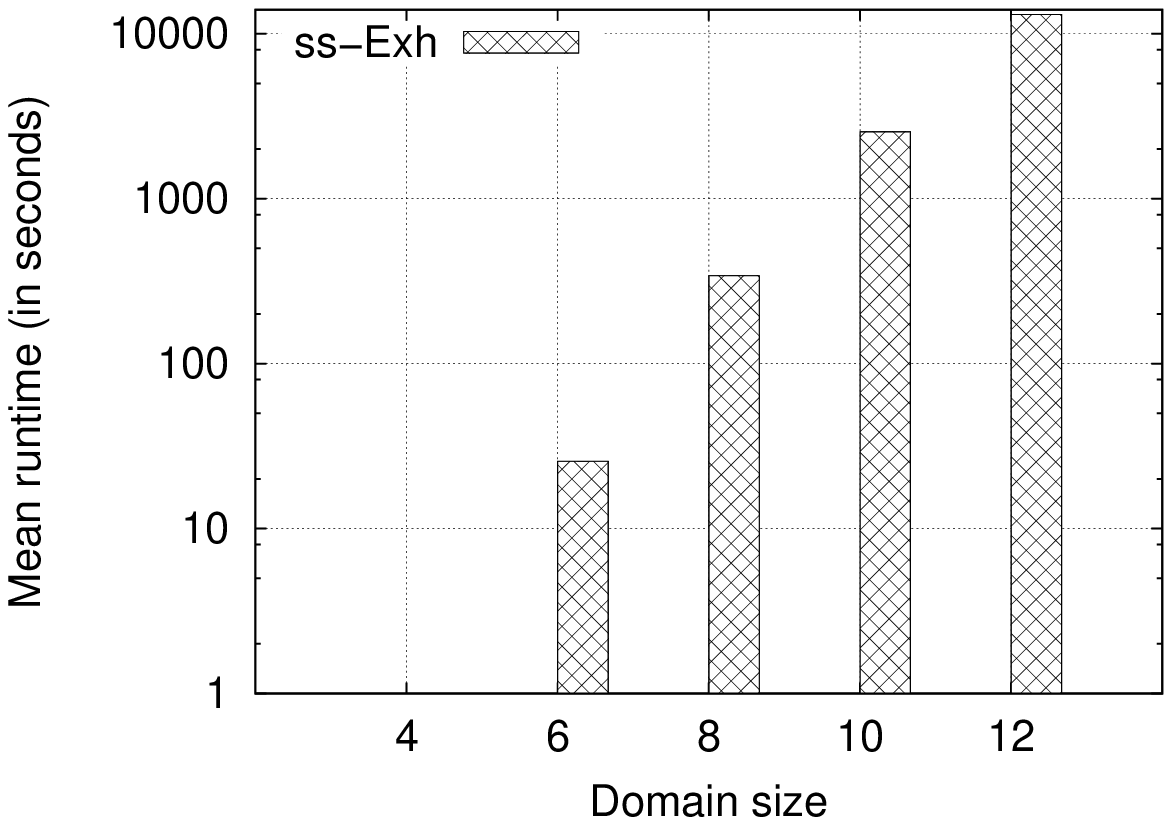}
}
\subfloat[Procurement]{
\includegraphics[width=0.45\linewidth,height=1.6in]{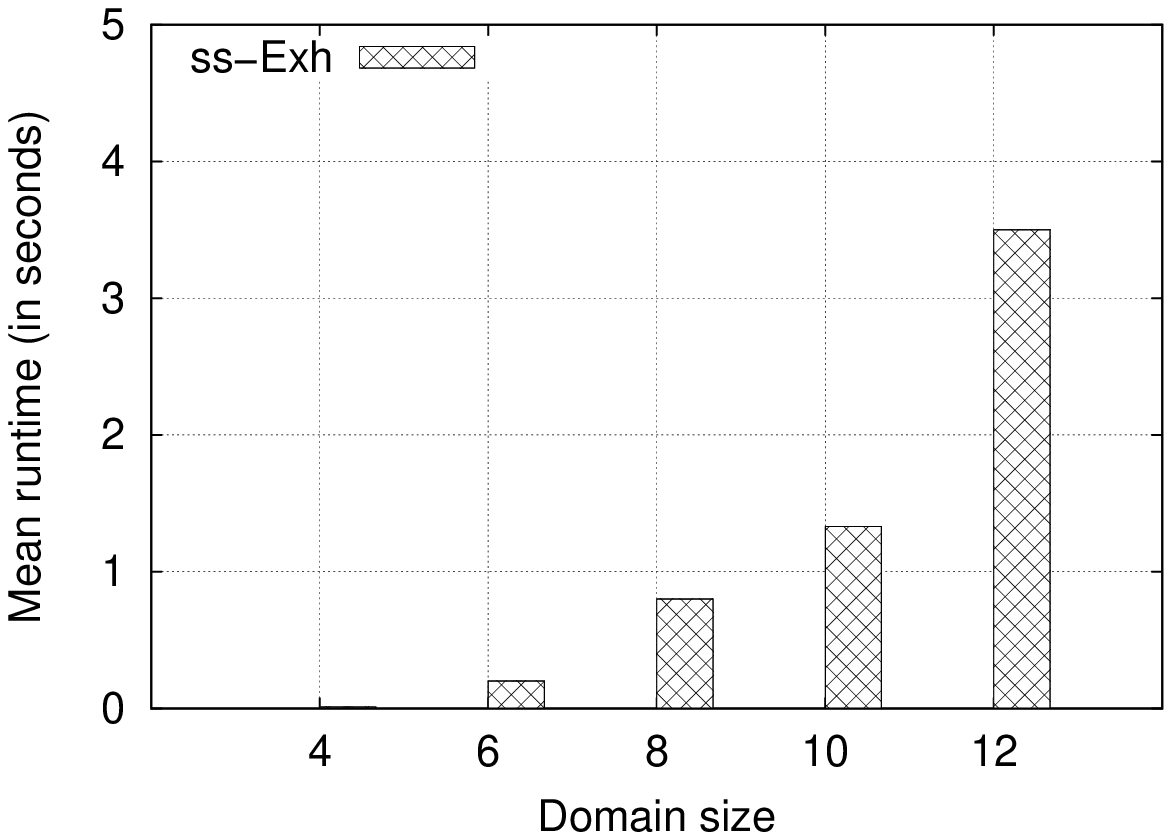}
}

\subfloat[Disbursement]{
\includegraphics[width=0.45\linewidth,height=1.6in]{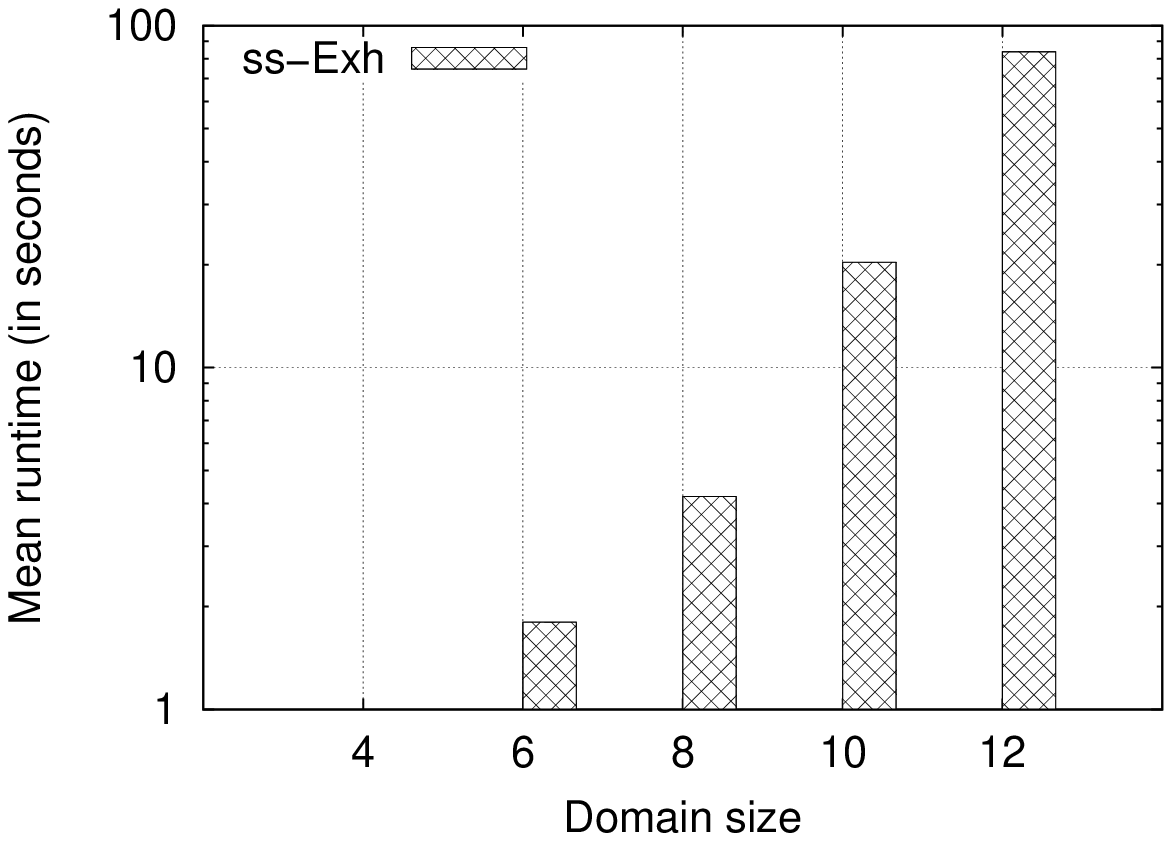}
}
\subfloat[Meal Options]{
\includegraphics[width=0.45\linewidth,height=1.6in]{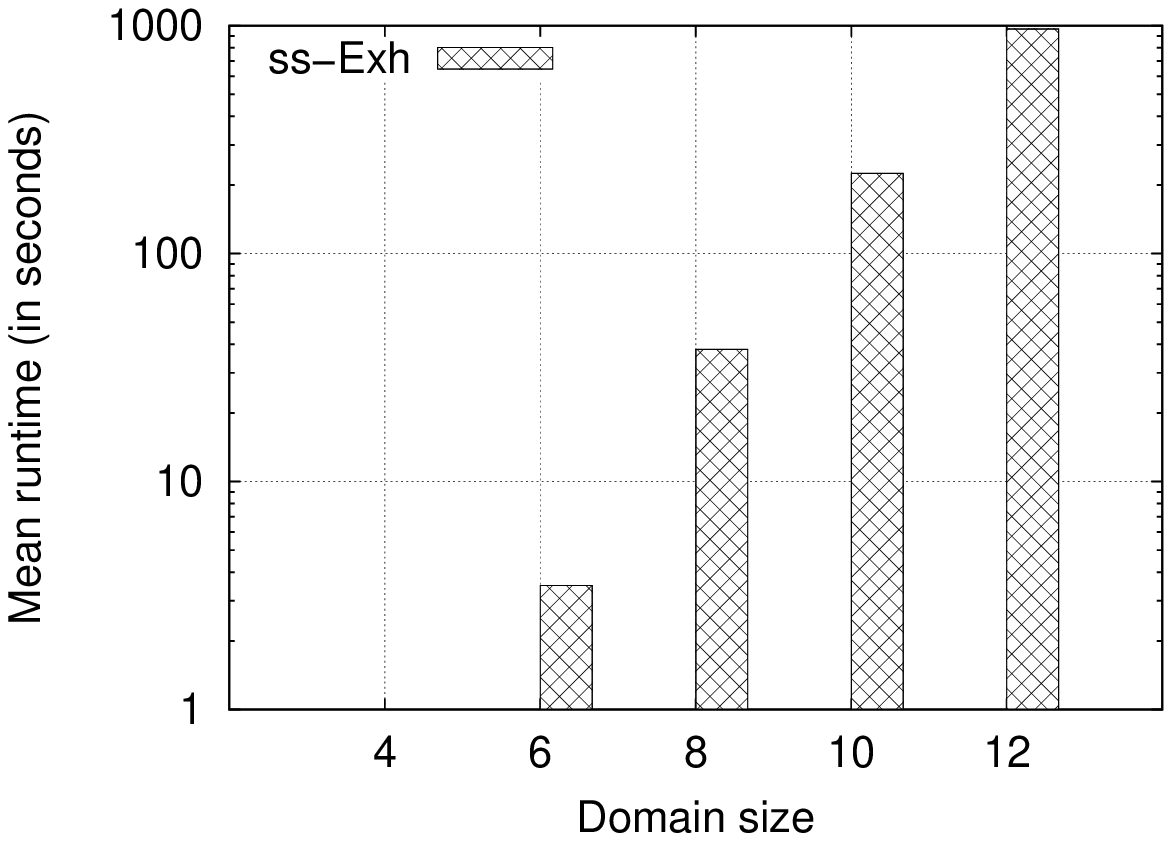}
}

\caption{Runtime of the exhaustive search on {\it simple} instances}
\label{Time-ss-Exh}
\end{figure}

In Figure~\ref{Time-ss-Exh}, we depict the mean runtime obtained from the exhaustive 
search algorithm ({\it ss-Exh}) on the class of {\it simple} instances. As one can notice, these 
runtime exponentially increase with the domain size which we recall, is the number of concrete 
services that can be assigned to an abstract operation. In the same example, the mean runtime {\bf in all experiments }
was lower than $2$ seconds for the backtracking based algorithms. The same trends were also observed in 
considering instances of the {\it medium} and {\it hard} class. This demonstrates that there is 
a large class of instances on which the backtracking based algorithms outperforms the exhaustive search. 

For validating our prior intuition regarding the considerable amount of useless work in exhaustive 
search (see Section~\ref{backtrackingSearch}), we quantified and computed this work. We defined the 
useless work of an algorithm as a complete assignment that does not improve the current 
local solution. In Figure~\ref{UselessExh}, we depict the useless work observed in some instances of the 
disbursement process. As one can notice, this quantity was always greater in exhaustive search. The 
same trend was observed in the other instances. 

\begin{figure}[ht]
\centering
\subfloat[Disbursement Medium]{
\includegraphics[width=0.45\linewidth,height=1.4in]{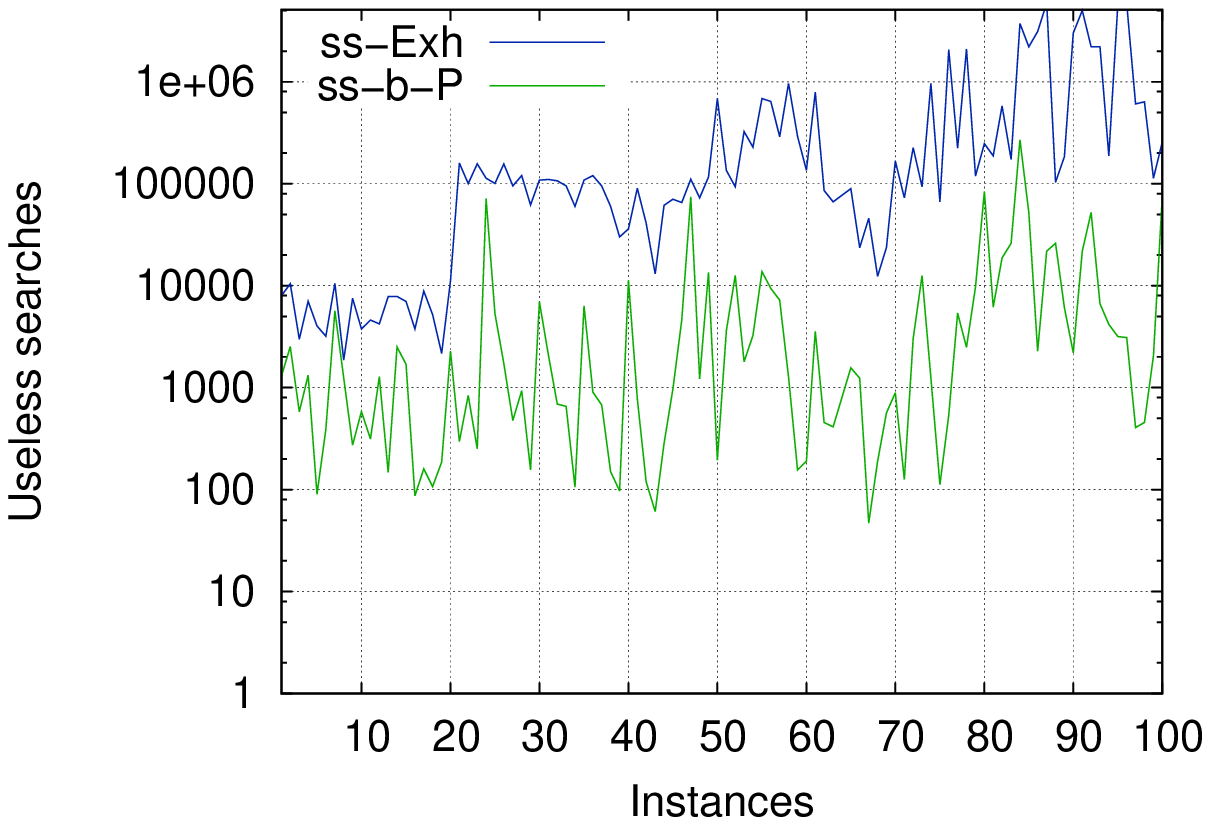}
}
\subfloat[Disbursement Hard]{
\includegraphics[width=0.45\linewidth,height=1.4in]{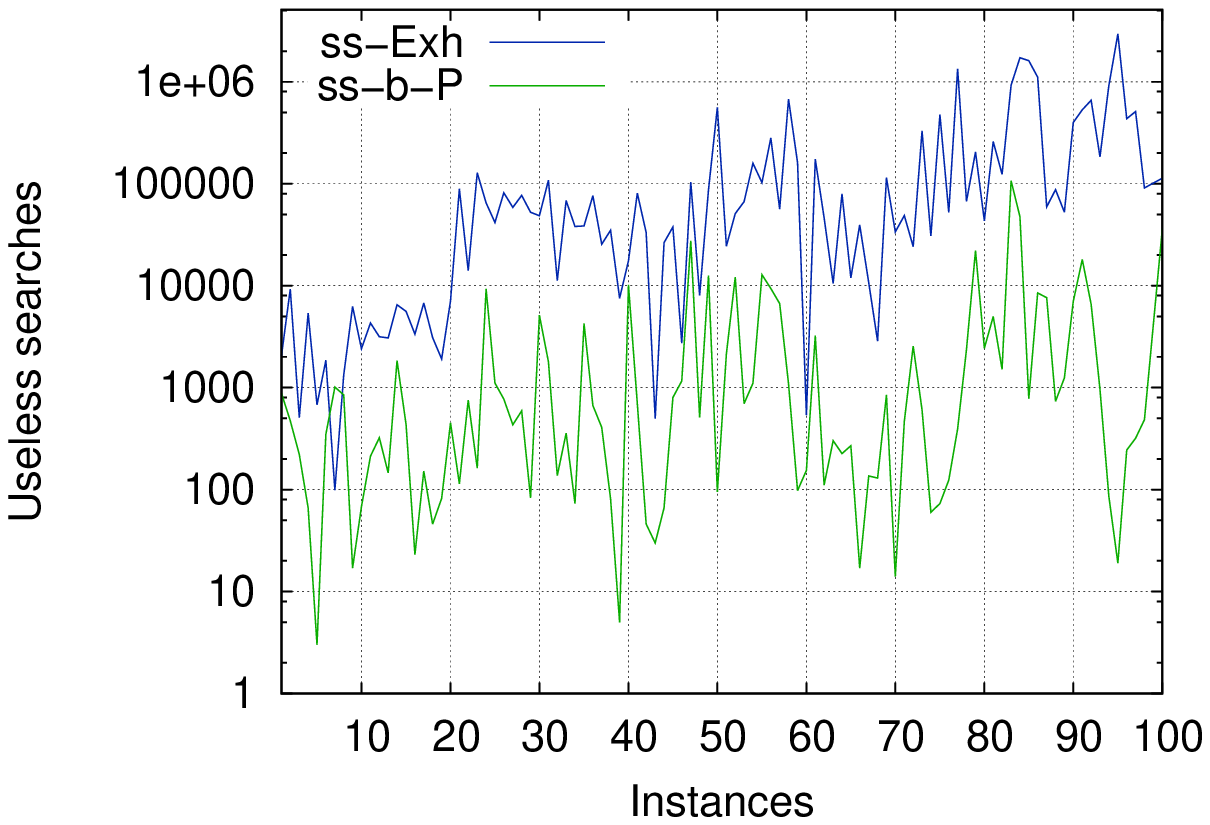}
}

\caption{Useless searches in exhaustive search and backtracking}
\label{UselessExh}
\end{figure}

\subsection{Backtracking versus Integer Programming}

We also compared the backtracking algorithms with integer linear programming. For this 
purpose, we used the integer modelling that we proposed in prior work~\cite{JISA}. This modelling 
was proposed for a more general class of services compositions. However, it supports our 
restricted setting. The integer model was run with the GLPK solver~\cite{GLPK}. In the search of 
the optimal solution, the solver internally computes a lagrangian relaxation for avoiding useless searches. 
Doing so, the run of our integer model is very close to the one of the BBLP algorithm.

In the first experiments, we compared the solvers and our approaches on the previously defined 
instances. We did not however see any significant differences in performances. We increased 
the domain sizes to $140$; but no significant differences appeared. The mean runtime was around $2.5$ seconds 
in either approach. For exhibiting runtime differences, we considered two other processes: the motif 
network and the genelife2 workflow. Both were taken from the Pegasus database~\cite{Pegasus} and come with 
different sizes (small, medium large). For both, we chose the small size variants. An illustrative representation 
of the chosen processes is given in Figure~\ref{Workflow}. 

\begin{figure}[ht]
\centering
\subfloat[Motif]{
\includegraphics[width=0.35\linewidth,height=1.4in]{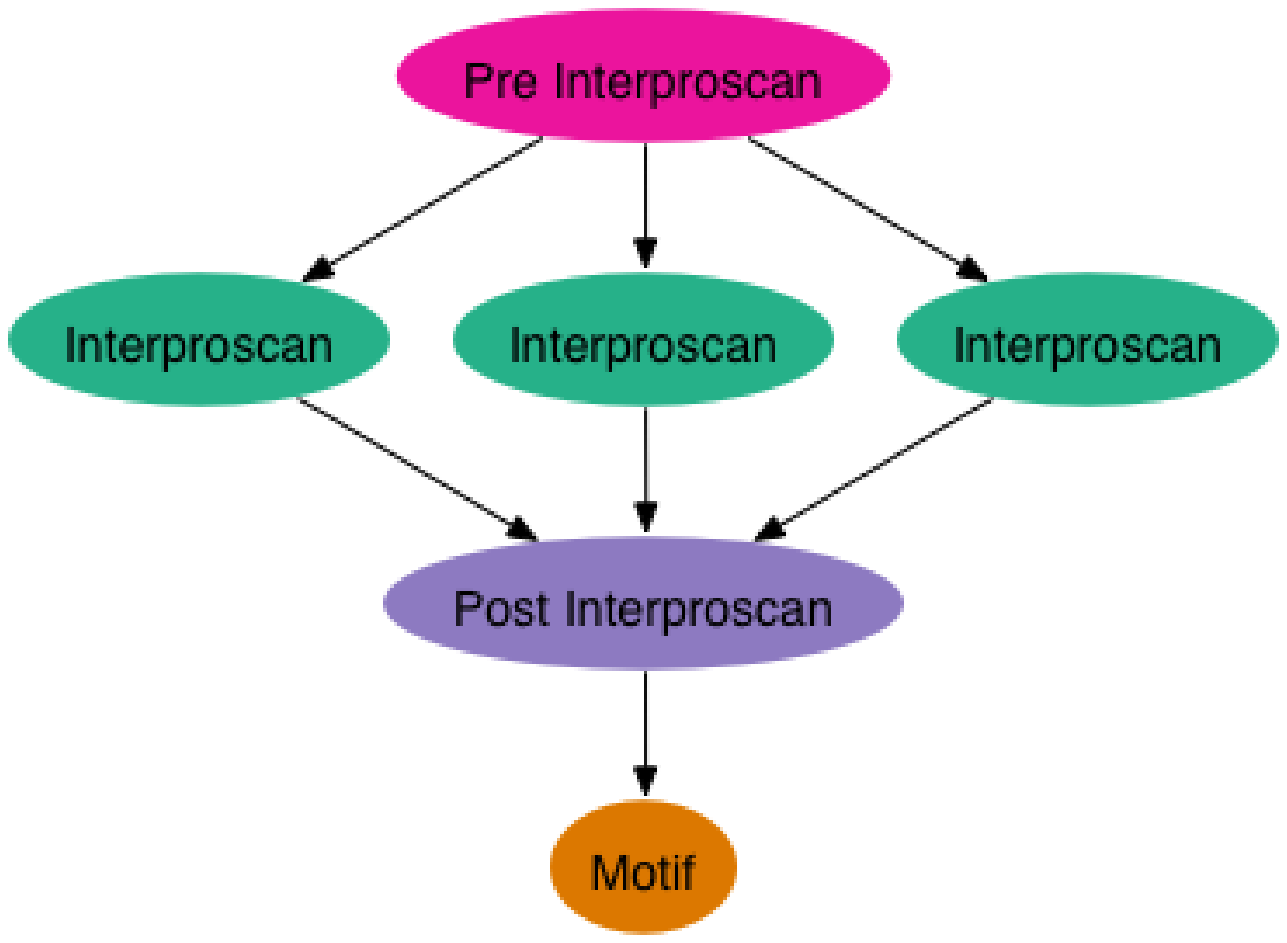}
}
\subfloat[Genelife2]{
\includegraphics[width=0.35\linewidth,height=1.4in]{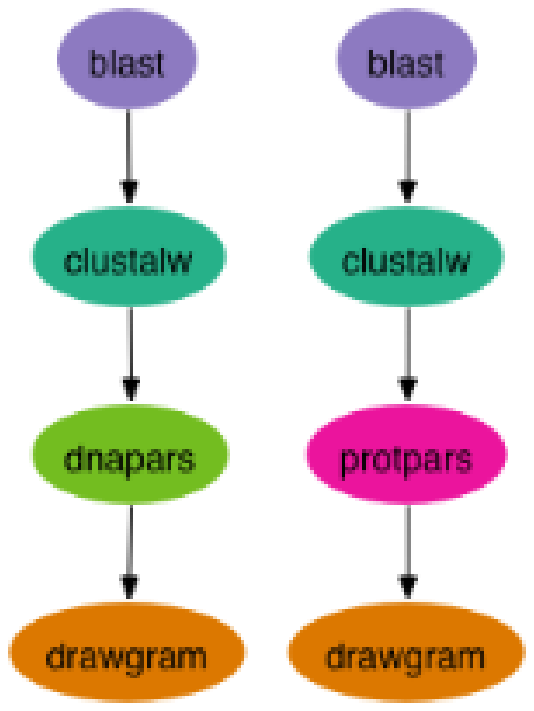}
}

\caption{Pegasus workflows}
\label{Workflow}
\end{figure}

On these two processes, we randomly generated $30$ instances using {\it simple} SLAs constraints and 
 $30$ ones with {\it hard constraints}. The domain size of the instances were taken between $4$, $8$, $16$, $32$, $64$

\begin{figure}[ht]
\centering
\subfloat[Genelife2 Hard]{
\includegraphics[width=0.45\linewidth,height=1.6in]{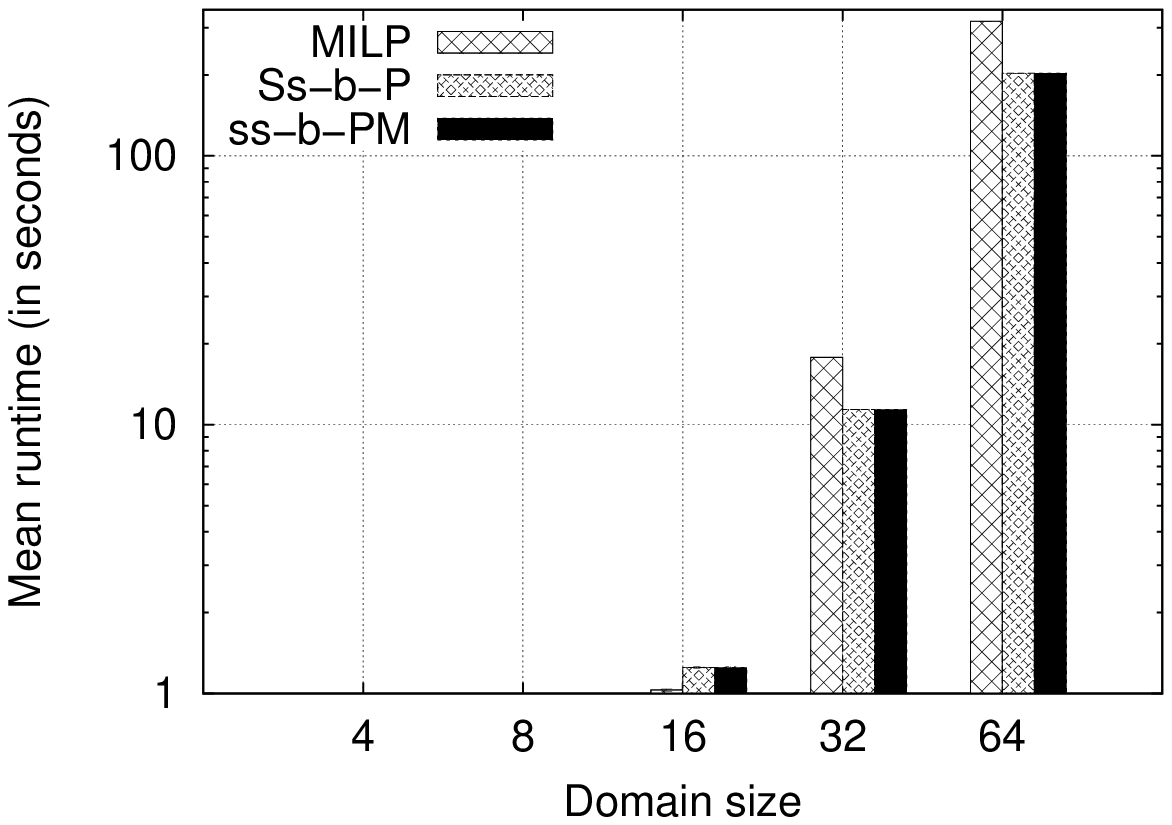}
}
\subfloat[Motif Hard]{
\includegraphics[width=0.45\linewidth,height=1.6in]{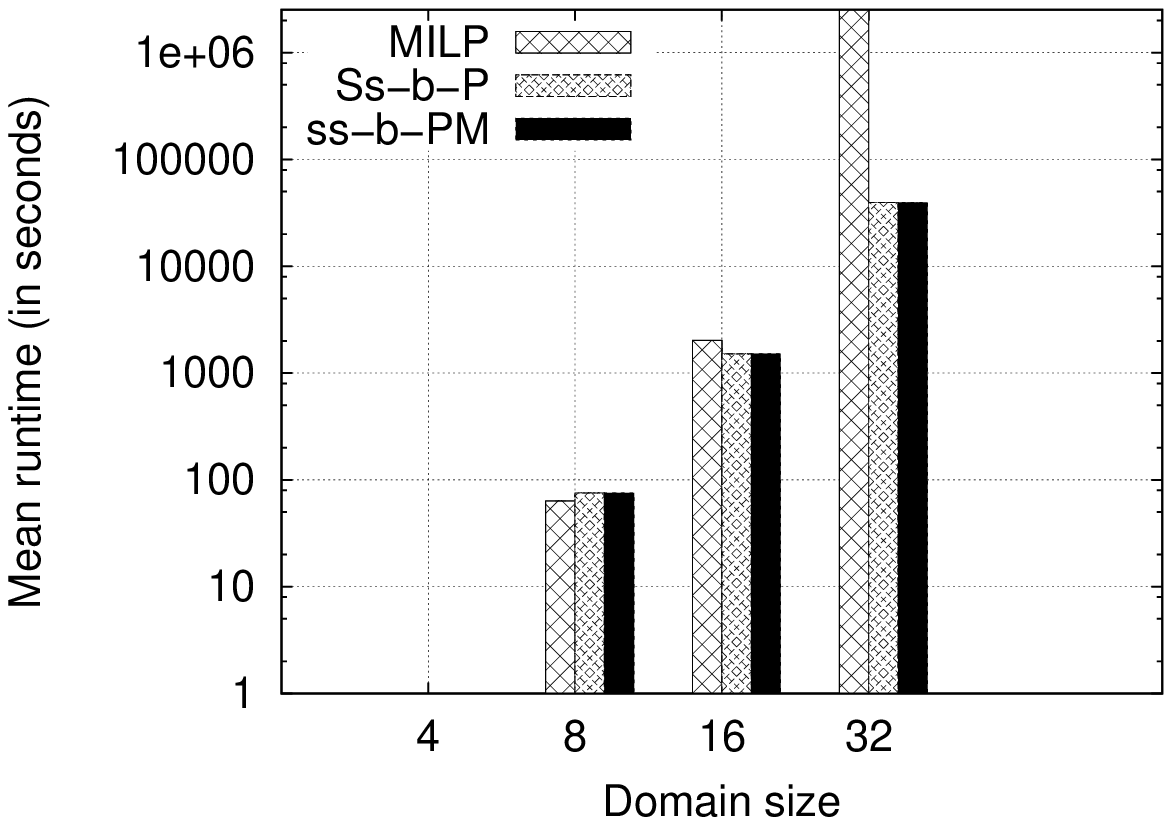}
}

\subfloat[Genelife2 Simple]{
\includegraphics[width=0.45\linewidth,height=1.6in]{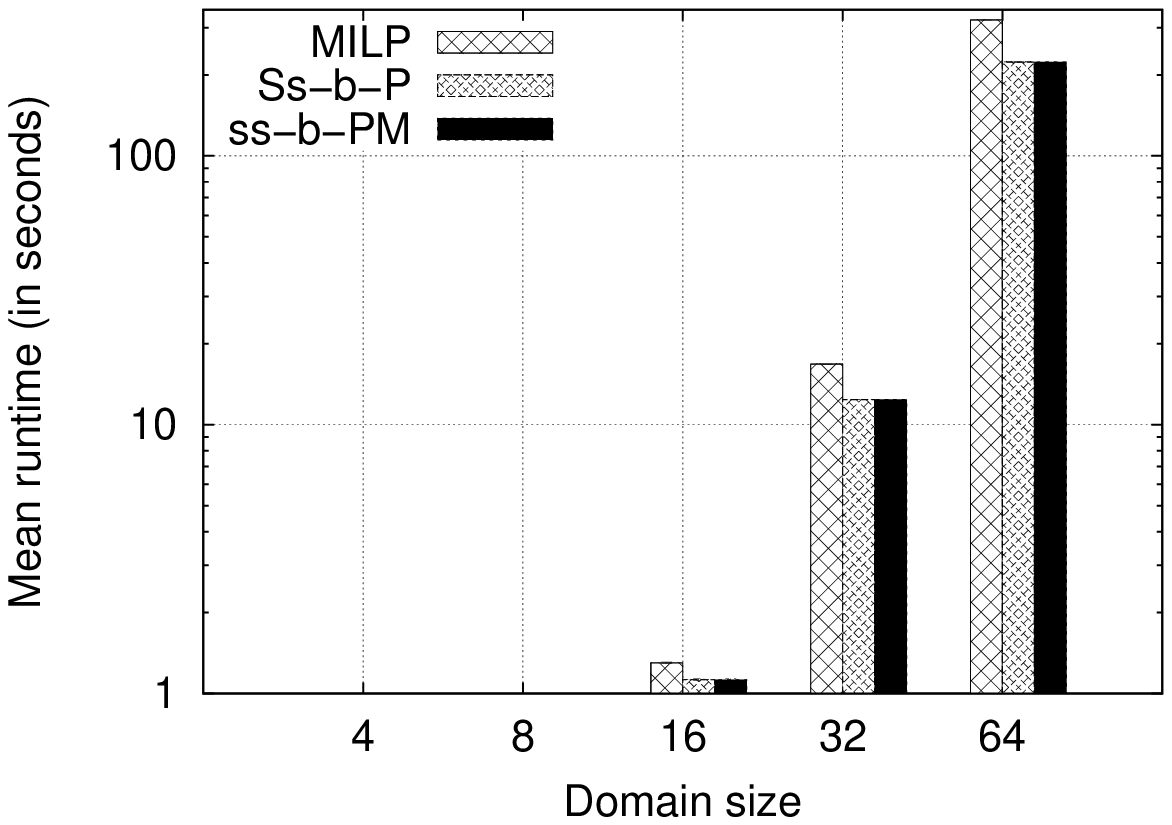}
}
\subfloat[Motif Simple]{
\includegraphics[width=0.45\linewidth,height=1.6in]{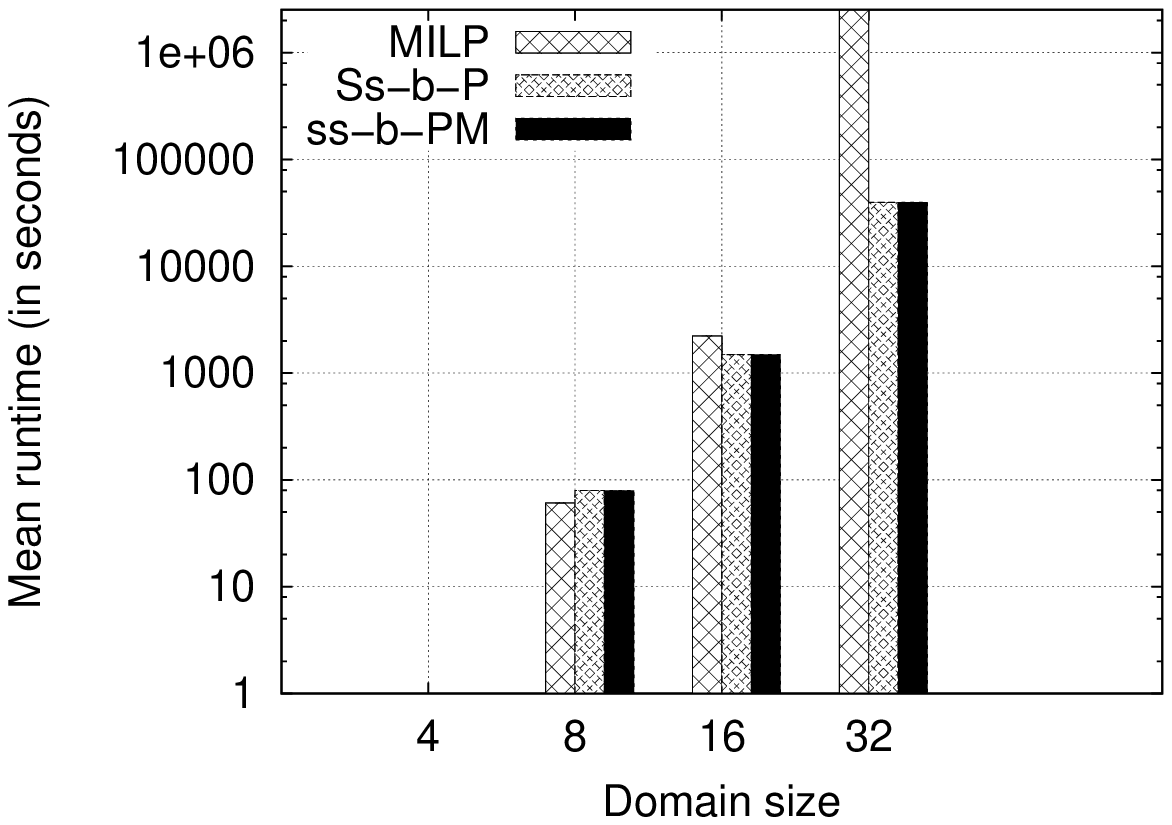}
}

\caption{Runtime of the exhaustive search on {\it simple} instances}
\label{Time-ss-Exh}
\end{figure}

The experimental results are depicted in Figure~\ref{Time-ss-Exh}. 
We noticed that the integer programming (MILP in the Figure) was dominated by the backtracking 
algorithms when we increase the domain size. In particular, for the motif network with a domain 
size equal to $32$, we were not able to get a solution from integer programming after $1$ week. 
We believe that differences between integer programming and backtracking were due to the 
more important number of operations of genelife2 and the motif network. Indeed, we finally had 
 $15$ operations nodes for the motif network and $11$ operations nodes for the genelife2 workflow. 

In addition to the superiority of backtracking, these experiments also revealed that if 
we can expect real-time results with backtracking in the case where we have less than $9$ 
abstract operations, with more than $10$ abstract operations, the algorithm becomes time 
consuming. But it is important to remark that given few number of abstract operations, we 
can have quick solutions even if the number of concrete services is important (greater than $1400$ 
for instance). 

\subsection{What is the best backtracking algorithm?}

The objective here was to determine what is the faster backtracking algorithm. 
On this point, our experiments did not reveal a clear trend. In Figure~\ref{Time-ss-Exh} 
for instance, one can notice that the algorithms are quite similar even if there are some runtime 
differences. We did additional experiments where instead of a fixed number of concrete services 
per abstract operations, we set a random number chosen each time between $3$ and a maximal 
domain size. In Figure~\ref{Time-ss-b}, we depict the results obtained on {\it hard }
instances. As one can notice, they do not define a particular trend. We believe that the small variations 
that we observed state that depending on the distribution of energy consumption and service response time, 
a backtracking algorithm can detect earlier some partial assignments that must not be complete.

\begin{figure}[ht]
\centering
\subfloat[Shipment Hard]{
\includegraphics[width=0.45\linewidth,height=1.6in]{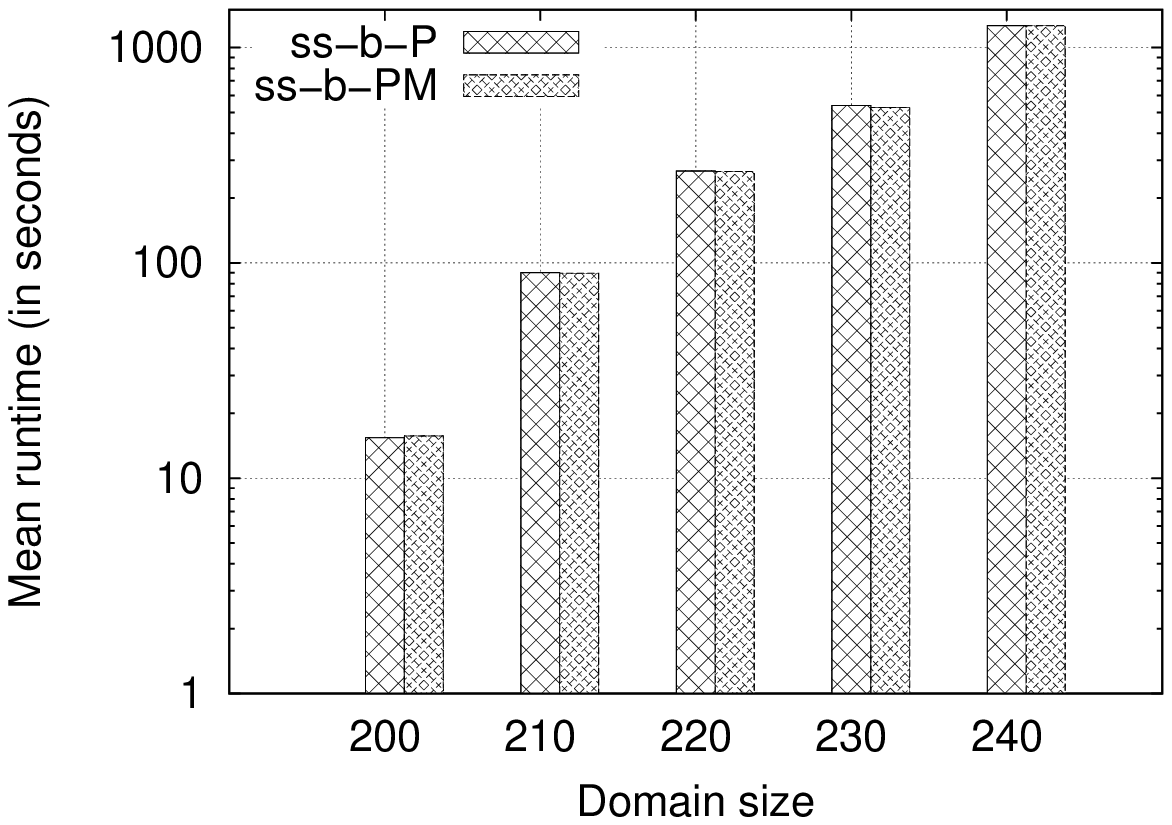}
}
\subfloat[Disbursement Hard]{
\includegraphics[width=0.45\linewidth,height=1.6in]{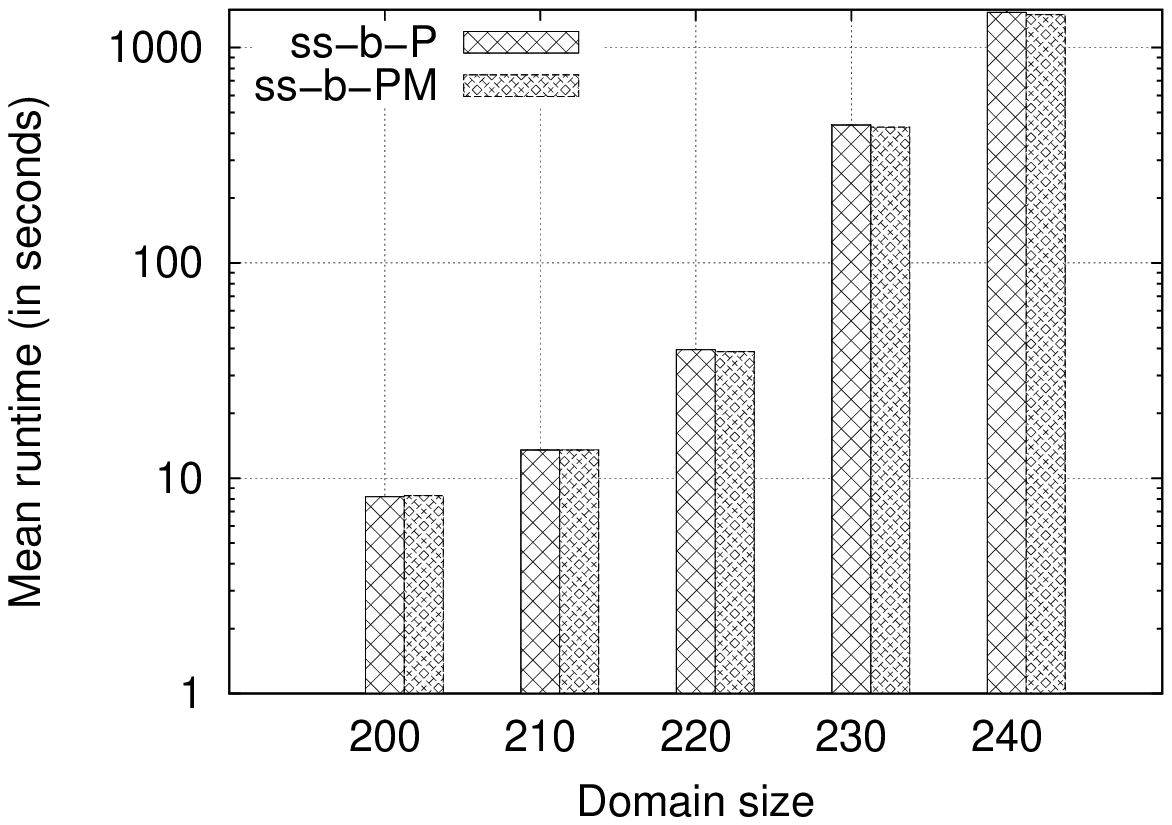}
}

\subfloat[Procurement Hard]{
\includegraphics[width=0.45\linewidth,height=1.6in]{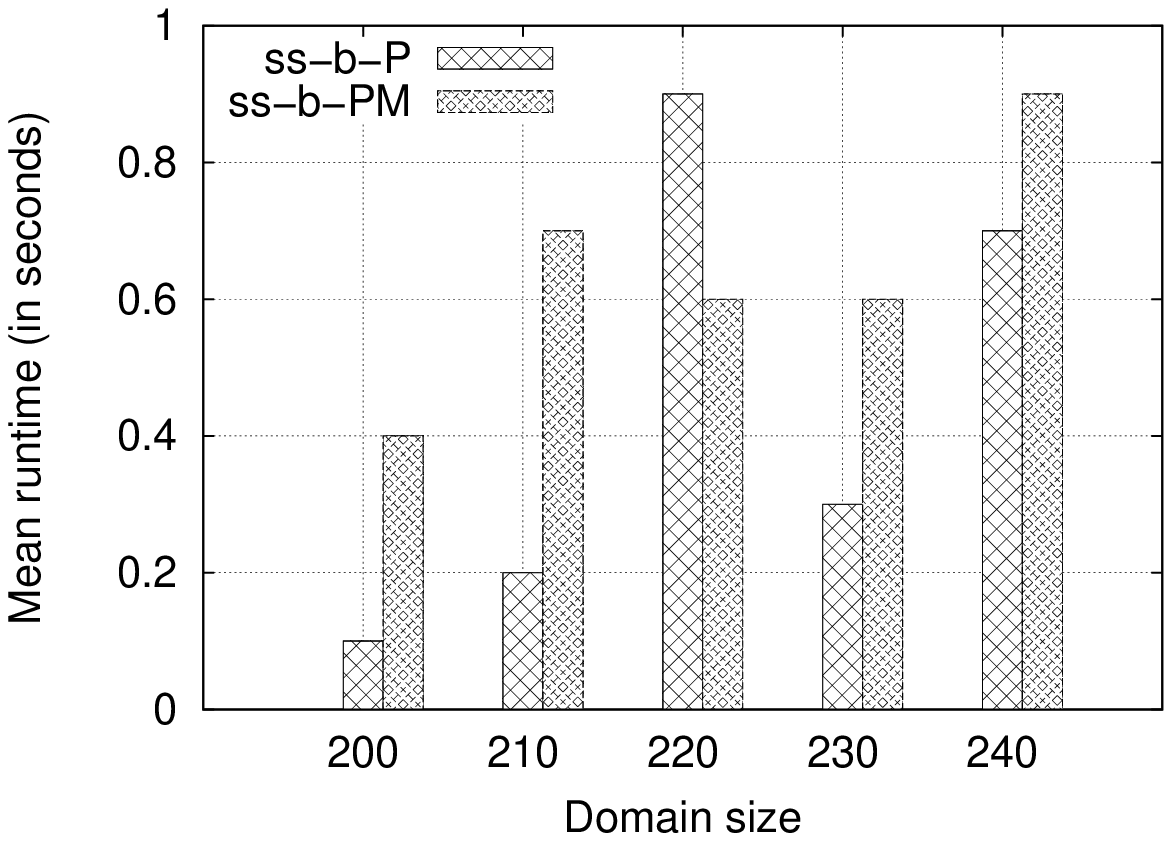}
}
\subfloat[Meal Hard]{
\includegraphics[width=0.45\linewidth,height=1.6in]{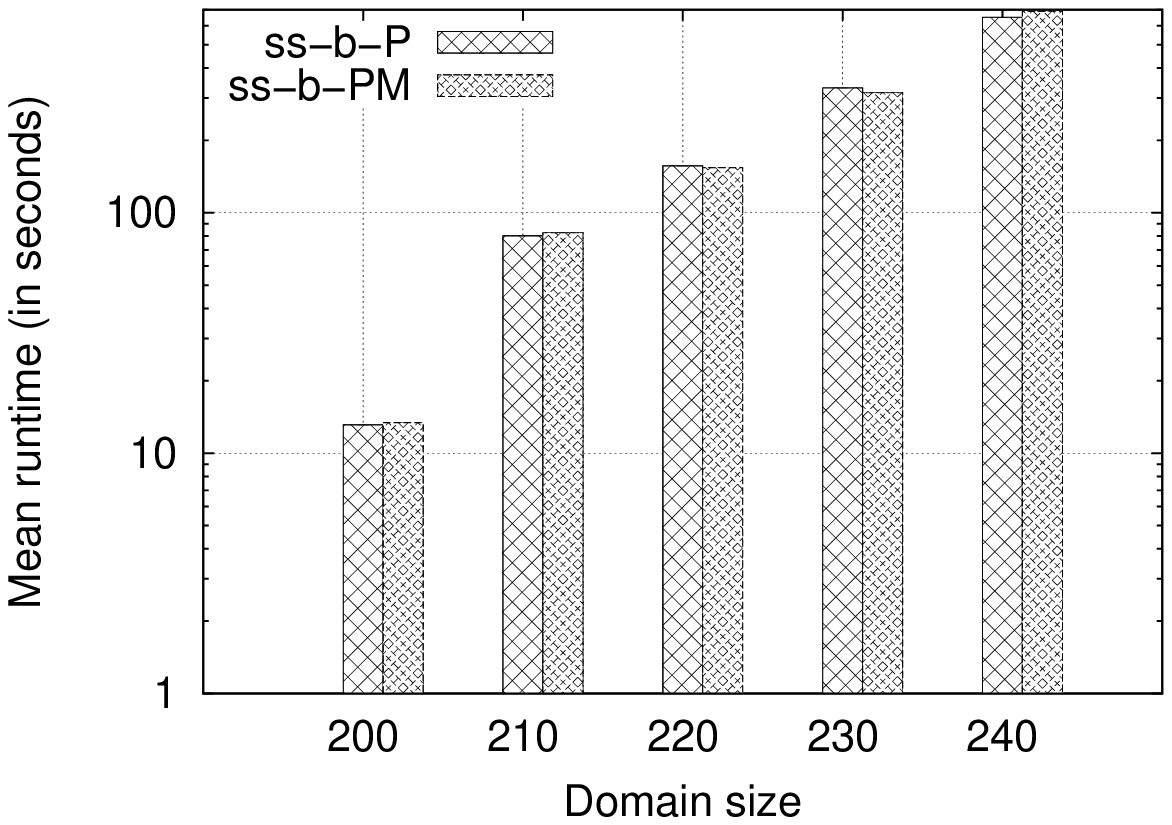}
}
\caption{Runtime of the backtracking algorithms on {\it hard} instances}
\label{Time-ss-b}
\end{figure}

Our experiments globally showed that the backtracking algorithms can outperform classical 
exact solutions for the service selection problem. But, we did not see any clear 
difference between the backtracking algorithms.

\section{Discussion} \label{Discussion}

In this paper, we proposed a novel approach for solving the service selection problem. 
The main idea here is to view this problem as a CSP. We now introduce a short discussion 
about the potential of this viewpoint.

As already mentioned, our proposal can also be adapted for the resolution of the 
feasibility problem. In an algorithmic viewpoint, it suffices in {\it SS-b-PM} to add 
a control that will stop the execution when a first solution is found. 

The experiments demonstrated that 
backtracking can be envisioned for real time service compositions on small processes. 
In the case of large processes, we propose to modify theses algorithms for obtaining quick solutions  
that are not necessarily optimal. A classical and simple idea for this is 
to include a cutoff time. Cutoff time means that in the algorithm, we include a control for returning 
the best computed solution when the cutoff time is reached. 

The CSP mapping that we propose can also inspire other algorithms for the service 
selection problem. While our proposal only focuses on the backtracking technique, it might
be interesting to consider other CSP resolution technique like the forward checking or the 
backjumping~\cite{Baker95intelligentbacktracking}. Moreover, we can also revisit our 
backtracking proposal to adapt it to other ordering techniques used in constraint 
satisfaction. We have for instance the max-domain, min-degree, max-domain/degree and 
the set of dynamic orderings. Finally there also exist many parallel algorithms for 
constraint satisfaction. The techniques employed for achieving parallelization can certainly 
be reused in the service selection problem.

\section{Conclusion} \label{Conclusion}

This paper proposes a novel approach for the resolution of the service selection problem. 
The main idea is to consider this problem as a particular case of the constraint satisfaction 
problem. We gave the theoretical mapping that supports this idea. We then derive from it two backtracking 
algorithms for service selection. We showed in the experiments that the proposed algorithms can 
outperform classical exact solutions for the service selection problem. However, for using them in 
real-time context, one must consider small services compositions. This result was expected due to the 
NP-hardness of the problem; however the techniques that we propose can drastically reduce the search 
space explored for finding the optimal services composition. 
 For continuing this work, we envision 
two main directions. In the first one, we are interesting in developing a parallel algorithm for the 
service selection problem based on what is done in CSP parallelization. Our second direction consists of 
applying our algorithms for service selection on real services compositions. 

\section*{Acknowledgments}

The experiments conducted in this work were conducted on the nodes B500 of the university of Paris 13 
Magi Cluster and available at http://www.univ-paris13.fr/calcul/wiki/

\bibliographystyle{hplain}
\bibliography{ServiceSelectionCSP-Report}

\end{document}